\documentclass[12pt]{article} 
\usepackage[utf8]{inputenc} 

\usepackage{graphicx} 
\usepackage{amsmath, amsfonts, amssymb, color, amsxtra}
\usepackage{bm}
\usepackage{enumerate}
\usepackage{lineno}
\usepackage[utf8]{inputenc}
\usepackage{verbatim}
\usepackage{appendix}
\usepackage{amsmath}
\usepackage{amsfonts}
\usepackage{amssymb}
\usepackage{amsthm}
\usepackage{graphicx}
\usepackage{natbib}
\usepackage{color}
\usepackage{caption}
\usepackage{multirow}
\usepackage[affil-it]{authblk}
\usepackage{algpseudocode, algorithm, algorithmicx}
\usepackage{xr}
\usepackage{tikz}
\usepackage{centernot}

\usepackage[top=1in, bottom=1in, left=0.8in, right=0.8in]{geometry} 
\usepackage{graphicx} 
\usepackage{amsmath, amsfonts, color}
\usepackage{bm}
\usepackage{enumerate}

\usepackage{booktabs} 
\usepackage{array} 
\usepackage{paralist} 
\usepackage{verbatim} 
\usepackage{subfig} 

\usepackage{fancyhdr} 
\usepackage{sectsty}
\allsectionsfont{\sffamily\mdseries\upshape} 

\usepackage[nottoc,notlof,notlot]{tocbibind} 
\usepackage[titles,subfigure]{tocloft} 

\usepackage{setspace}
\newtheorem{theorem}{Theorem}[section]
\newtheorem{lemma}[theorem]{Lemma}
\newtheorem{proposition}[theorem]{Proposition}

\newcommand{\bb}{{\bm b}}

\newcommand{\bv}{{\bm v}}
\newcommand{\bw}{{\bm w}}
\newcommand{\bx}{{\bm x}}

\newcommand{\bX}{{\bm X}}

\newcommand{\bbeta}{{\bm \beta}}

\newcommand{\btau}{{\bm \tau}}
\newcommand{\bSigma}{{\bm \Sigma}}

\newcommand{\bOmega}{{\bm \Omega}}

\newcommand{\bepsilon}{{\bm \epsilon}}
\newcommand{\bone}{{\bm 1}}
\newcommand{\bzero}{{\bm 0}}

\newcommand{\cI}{\mathcal{I}}

\newcommand{\cE}{\mathcal{E}}
\newcommand{\cV}{\mathcal{V}}

\newcommand{\E}{\mathbb{E}}
\newcommand{\Var}{\mathrm{Var}}

\DeclareMathOperator*{\argmin}{arg\,min}
\DeclareMathOperator*{\argmax}{arg\,max}

\newcommand{\supp}{\mathrm{supp}}
\newcommand{\diag}{\mathrm{diag}}

\newcommand{\RN}[1]{%
  \textup{\uppercase\expandafter{\romannumeral#1}}%
}
\newcommand{\independent}{\mathrel{\perp\mspace{-10mu}\perp}}
\newcommand{\nindependent}{\centernot{\independent}}

\newcommand\smallO{
  \mathchoice
    {{\scriptstyle\mathcal{O}}}
    {{\scriptstyle\mathcal{O}}}
    {{\scriptscriptstyle\mathcal{O}}}
    {\scalebox{.7}{$\scriptscriptstyle\mathcal{O}$}}
  }

\date{} 

\author{Sen Zhao\thanks{senzhao@google.com.}} 
\author{Stephen Ottinger}
\author{Suzanne Peck}
\author{Christine Mac Donald}
\author{Ali Shojaie}
\affil{Google Research, Seattle Children's Hospital and University of Washington}

\title{Network Differential Connectivity Analysis}

\begin{document}

\maketitle
\def\spacingset#1{\renewcommand{\baselinestretch}%
{#1}\normalsize} \spacingset{1}
\abstract{
Identifying differences in networks has become a canonical problem in many biological applications. Here, we focus on testing whether two Gaussian graphical models are the same. Existing methods try to accomplish this goal by either directly comparing their estimated structures, or testing the null hypothesis that the partial correlation matrices are equal. However, estimation approaches do not provide measures of uncertainty, e.g., $p$-values, which are crucial in drawing scientific conclusions. On the other hand, existing testing approaches could lead to misleading results in some cases. To address these shortcomings, we propose a \emph{qualitative}  hypothesis testing framework, which tests whether the connectivity patterns in the two networks are the same. Our framework is especially appropriate if the goal is to identify nodes or edges that are differentially connected. No existing approach could test such hypotheses and provide corresponding measures of uncertainty, e.g., $p$-values. We investigate theoretical and numerical properties of our proposal and illustrate its utility in biological applications. Theoretically, we show that under appropriate conditions, our proposal correctly controls the type-I error rate in testing the qualitative hypothesis. Empirically, we demonstrate the performance of our proposal using simulation datasets and applications in cancer genetics and brain imaging studies. 
}

\spacingset{1.5}

\section{Introduction}\label{sec:intro}
Changes in biological networks, such as gene regulatory and brain connectivity networks, have been found to associate with the onset and progression of complex diseases \citep[see, e.g.,][]{BassettBullmore2009network, Barabasietal2011network}. Locating differentially connected nodes in the network of diseased and healthy individuals---referred to as \emph{differential network biology} \citep{IdekerKrogan2012}---can help researchers delineate underlying disease mechanism. Such \emph{network-based biomarkers} can also serve as effective diagnostic tools and guide new therapies. In this paper, we propose a novel inference framework, \emph{differential network analysis}, for identifying differentially connected nodes or edges in two networks.

Let $ne_j^m$ be the neighborhood of node $j$ in network $G^m$, i.e.,
\begin{align}
ne_j^m\equiv\{k\neq j:(j,k)\in\cE^m\}, \quad m\in\{\RN{1}, \RN{2}\}.
\end{align} 
The scientists' quest to identify differences in the two networks corresponds to testing $H^\ast_{0,j}: ne_j^{\RN{1}}= ne_j^{\RN{2}}$ versus $H^\ast_{a,j}: ne_j^{\RN{1}}\neq ne_j^{\RN{2}}$. Under $H^\ast_{0,j}$, node $j$ is connected to the same set of nodes in both networks. 

Of course, we do not directly observe the networks; rather, we observe noisy data $\bX^{\RN{1}}$ and $\bX^{\RN{2}}$ that are generated based on the underlying networks. Let $\bOmega^m$ be the inverse population covariance matrix of $\bX^m$, also known as the precision matrix. With Gaussian graphical model, nodes $j, k\in\cV$ are connected in network $G^m$ if and only if $\Omega_{jk}^m\neq0$; this quantity is proportional to the partial correlation between $\bx_j^m$ and $\bx_k^m$.  Thus, we can recast $H^\ast_{0,j}$ and $H^\ast_{a,j}$ as the following equivalent hypotheses
\begin{align}
H_{0,j}: \supp\left(\bOmega^\RN{1}_j\right)= \supp\left(\bOmega^\RN{2}_j\right), \label{eq:hypothesis}\\
H_{a,j}: \supp\left(\bOmega^\RN{1}_j\right)\neq \supp\left(\bOmega^\RN{2}_j\right),
\end{align}
where $\bOmega_j$ denotes the $j$th column of the precision matrix $\bOmega$. While our results are also valid with sub-Gaussian data, without Gaussianity, the network encodes conditional correlation, rather than conditional dependence; the interpretation of test results would thus change. 

The problem of testing differential connectivity in two networks has attracted much attention recently, and many approaches have been proposed to examine the equality of the \emph{values} in two precision matrices. However, as shown in Section~\ref{sec:challenges}, \emph{quantitative} inference procedures focused on values of the precision matrices may lead to misleading conclusions about \emph{structural differences} in two networks. In contrast, our proposal directly examines the \emph{support} of two precision matrices. 
This \emph{qualitative} testing framework, which we call \emph{differential connectivity analysis} (DCA), is specifically designed to address this challenge and is directly focused on the goal of identifying differential connectivity in biological networks. To the best of our knowledge, DCA is the first inference framework that can formally test \emph{structural differences} in two networks, i.e., $H^\ast_{0,j}: ne_j^{\RN{1}}= ne_j^{\RN{2}}$. Moreover, DCA is a general framework that can  incorporate various estimation and hypothesis testing methods for flexible implementation and easy extensibility.

\subsection{Related Work}\label{sec:challenges}

In this section, we summarize related work and discuss why existing approaches are unable to test $H^\ast_{0,j}: ne_j^{\RN{1}}= ne_j^{\RN{2}}$.

In most applications, the edge sets $\cE^{\RN{1}}$ and $\cE^{\RN{2}}$ are estimated from data, based on similarities/dependencies between variables. In particular, Gaussian graphical models (GGMs) are commonly used to estimate biological networks \citep[e.g.,][]{Krumsieketal2011}. To identify differential connectivities in two networks, we may na\"ively eyeball the differences in two GGMs estimated using single network estimation methods \citep[e.g., ][]{MeinshausenBuhlmann2006, Friedmanetal2008GL}, or joint estimation methods \citep[e.g.,][]{Guoetal2011JEG, Danaheretal2014JGL, Zhaoetal2014JEG, Petersonetal2015bayesiangraph, SaegusaShojaie2016}. However, these estimation approaches do not provide measures of uncertainty, e.g., $p$-values, and are thus of limited utility for drawing scientific conclusions. 

Building upon network estimation methods, a number of recent approaches provide confidence intervals and/or $p$-values for high-dimensional precision matrices. The first class of hypothesis testing procedures focuses on a single precision matrix \citep{Renetal2015GGM, JankovavadeGeer2015glasso, JankovavandeGeer2017, XiaLi2017}. These methods examine the null hypothesis $\Omega^m_{jk}=0, m\in\{\RN{1},\RN{2}\}$ for $j\neq k$, and hence could control the probability of falsely detecting an nonexistent edge. However, they could not control the false positive rates of $H^\ast_{0,j}$. This is because $H^\ast_{0,j}$ concerns the \emph{coexistence} of edges in two networks. Thus, the false positive rate of $H^\ast_{0,j}$ not only depends on probability of falsely detecting an nonexistent edge, but also depends on the probability of correctly detecting an existent edge. While single network hypothesis testing methods control the former probability, they do not control the latter.

The second class of inference procedures examines whether corresponding entires in two precision matrices are equal. For example, \citet{Xiaetal2015pt} tests whether $\Omega_{jk}^\RN{1}=\Omega_{jk}^\RN{2}$, \citet{Belilovskyetal2016diffprec} tests whether $\Omega_{jk}^\RN{1}/\Omega_{jj}^\RN{1}=\Omega_{jk}^\RN{2}/\Omega_{jj}^\RN{2}$, while \citet{StadlerMukherjee2016} tests whether $\phi^\RN{1}=\phi^\RN{2}$, where $\phi$ parametrizes the underlying data generation distribution. Permutation based methods have also been proposed in \citet{Gilletal2014}. The primary limitation of these methods is that examining differences in magnitudes of $\bOmega^{\RN{1}}$ and $\bOmega^{\RN{2}}$ may lead to misleading conclusions.
Consider the following toy example with three Gaussian variables: suppose in population ${\RN{1}}$, variable 1 causally affects variables 2 and 3, and variable 2 causally affects variable 3. 
Suppose, in addition, that in population ${\RN{2}}$, the effect of variable 1 on variable 2 remains intact, while the effect of variables 1 and 2 on variable 3 no longer present due to, e.g., a mutation in the latter. 
The undirected networks corresponding to the two GGMs are portrayed in Figure~\ref{fig:quanvsqual}.

\begin{figure}[h]
\caption{Conditional dependency structures of variables in populations ${\RN{1}}$ and ${\RN{2}}$.} 
\vspace{1em}
\label{fig:quanvsqual}
\centering
\includegraphics[width = 0.5\textwidth]{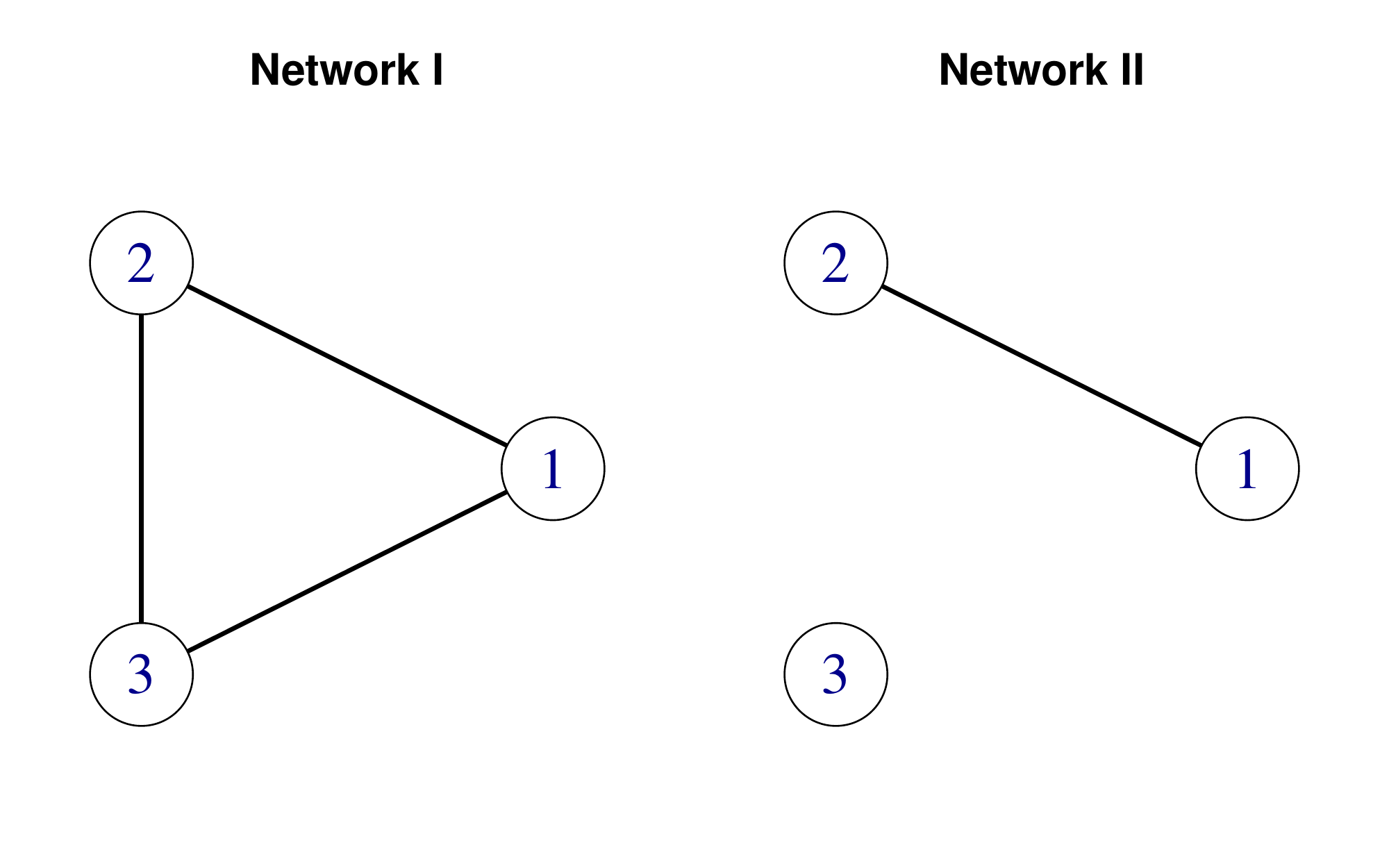}
\end{figure}

Suppose, without loss of generality, that the precision matrix of variables in population $\RN{1}$ is 
\[
\bOmega^{\RN{1}}=\begin{bmatrix}
    1       & 0.5 & 0.5 \\
    0.5    & 1    & 0.5 \\
    0.5    & 0.5    & 1
\end{bmatrix}.
\]
Further, suppose that $\bx_1^{\RN{2}}$ has the same distribution as $\bx_1^{\RN{1}}$, i.e., $\bx_1^{\RN{2}}\sim_d\bx_1^{\RN{1}}$. 
The unchanged (causal) relationship of variables 1 and 2 leads to $\bx_2^{\RN{2}}\sim_d\bx_2^{\RN{1}}$. On the other hand, $\bx_3^{\RN{2}}$ is independent of $\bx_1^{\RN{2}}$ and $\bx_2^{\RN{2}}$, i.e., $\bx_3^{\RN{2}}\independent\bx_{\{1,2\}}^{\RN{2}}$. Assuming, for simplicity, that $\Var\left(\bx_3^{\RN{2}}\right)=1$, we can verify that (see Section~\ref{sec:quanvsqual})
\[
\bOmega^{\RN{2}}=\begin{bmatrix}
    0.75       & 0.25 & 0 \\
    0.25    & 0.75    & 0 \\ 
    0    & 0    & 1
\end{bmatrix}.
\]

In this example, the relationship between variables 1 and 2 is the same in both populations. In particular, the dependence relationship between $\bx_1$ and $\bx_2$ is the same, as indicated in Figure~\ref{fig:quanvsqual}. However, $\Omega_{1,2}^\RN{1}\neq \Omega_{1,2}^\RN{2}$ and $\Omega_{1,2}^\RN{1}/\Omega_{1,1}^\RN{1}\neq \Omega_{1,2}^\RN{2}/\Omega_{1,1}^\RN{2}$. Thus, existing quantitative tests \citep{Gilletal2014, Xiaetal2015pt, Belilovskyetal2016diffprec, StadlerMukherjee2016} would falsely detect $(1,2)$ as a differentially connected edge. 

At a first glance, the differences between quantitative and qualitative inference procedure may seem negligible. In fact, one may wonder whether the phenomenon demonstrated in the above toy example would manifest to meaningful false positive errors in more realistic settings with larger networks. To illustrate that quantitative tests may fail to control the type-I error rate of qualitative hypotheses $H^\ast_{0, j}: ne^{\RN{1}}_j=ne^{\RN{2}}_j$, we examined how the permutation test of \citet{Gilletal2014} controls the type-I error rate for $H^\ast_{0, j}$ in the simulation setting of Section~\ref{sec:sims}. The type-I error rates for node-wide tests of differential connectivity are shown in Table~\ref{tab:dna}. 

\begin{table}[t]
\centering
\caption{\footnotesize Type-I error rates for the test of \citet{Gilletal2014} under the settings of Figure~\ref{fig:simGraceI}}
\label{tab:dna}
\begin{tabular}{c|cccc}
\hline\hline
Sample Size & n = 100 &n = 200 & n = 400 & n = 800\\
Type-I Error Rate & 0.998 & 0.995 & 0.981 & 0.919\\
\hline
\end{tabular}
\end{table}

The errors in Table~\ref{tab:dna} seem unbelievably large. But note that even on qualitatively identically connected nodes $j\in\cV: ne^{\RN{1}}_j=ne^{\RN{2}}_j$, the connection strength as reflected by partial correlation may be vastly different on \emph{many} edges due to the difference in qualitative connectivity of \emph{other} nodes. We further illustrate that quantitative tests do not properly control the type-I error rate using real gene expression data in Section~\ref{sec:RD}.
As a result, any quantitative test with sufficient power will falsely conclude that those identically connected nodes are differentially connected. That is exactly the issue highlighted in the above toy example: Tests for equality of (partial) correlations values \citep[e.g.,][]{Gilletal2014, Xiaetal2015pt, Belilovskyetal2016diffprec,StadlerMukherjee2016} were designed to identify quantitative differences, and should not be used to identify qualitative differences in the two networks, or differential connectivity, if that is indeed the scientific question.

\section{Differential Connectivity Analysis}\label{sec:methods}

\subsection{Summary of the Proposed Framework}\label{sec:summary}
In this subsection, we present a high-level summary of the proposed differential connectivity analysis (DCA) framework. Details are provided in the following   subsections. 

Consider a node $j\in\cV$. Then, any other node $k\neq j$ must belong to one of three categories, which are depicted in Figure~\ref{fig:venn}a:
\begin{enumerate}
	\item[i)] $k$ is a common neighbor of $j$, i.e., $k\in ne_j^{\RN{1}}\cap ne_j^{\RN{2}}\equiv ne_j^0$;
	\item[ii)] $k$ is a neighbor of $j$ in one and only one of the two networks, i.e., $k\in ne_j^{\RN{1}}\,\triangle\,ne_j^{\RN{2}}$, where ``$\triangle$'' is the symmetric difference operator;  
	\item[iii)] $k$ is not a neighbor of $j$ in either network, i.e., $k\notin ne_j^{\RN{1}}\cup ne_j^{\RN{2}}$. 
\end{enumerate}
Clearly, $ne_j^{\RN{1}}=ne_j^{\RN{2}}$ implies $ne_j^{\RN{1}}\,\triangle\,ne_j^{\RN{2}}=\varnothing$. If, to the contrary, there exists a node $k$ such that $k\in ne_j^{\RN{1}}\,\triangle\,ne_j^{\RN{2}}$, then $j$ is differentially connected, i.e., $ne_j^{\RN{1}}\neq ne_j^{\RN{2}}$.

Thus, to test $H^\ast_{0,j}: ne_j^{\RN{1}}= ne_j^{\RN{2}}$, we propose to examining whether there exists a node $k$ such that $k\in ne_j^{\RN{1}}\,\triangle\,ne_j^{\RN{2}} \equiv (ne_j^{\RN{1}}\cup ne_j^{\RN{2}})\backslash(ne_j^{\RN{1}}\cap ne_j^{\RN{2}})$. Specifically, for a $k\neq j$ such that $k\notin ne_j^{\RN{1}}\cap ne_j^{\RN{2}}\equiv ne_j^0$, we check whether $k\in ne_j^{\RN{1}}\cup ne_j^{\RN{2}}$. 

In practice, we do not observe $ne_j^0$ and need to estimate it. Our hypothesis testing framework thus consists of two steps:
\begin{enumerate}
	\item \emph{Estimation}: We estimate the common neighbors of each node $j$ in the two networks, $ne_j^0\equiv ne_j^{\RN{1}}\cap ne_j^{\RN{2}}$; this estimate is denoted by $\widehat{ne}_j^0$. 
	\item \emph{Hypothesis Testing}: We test whether there exists a $k\notin \widehat{ne}_j^0$ such that $k\in ne_j^{\RN{1}}\cup ne_j^{\RN{2}}$. 
\end{enumerate}
Details of the above two steps are described in the next two subsections, where it becomes clear that the procedure can be naturally extended to test differential connectivity in more than two networks. 

From the discussion in the following subsections, it will also become evident that the estimated common neighborhood $\widehat{ne}_j^0$ plays an important role in the validity and power of the proposed framework. In Section~\ref{sec:est}, we show that in order for $\widehat{ne}_j^0$ to be useful in the hypothesis testing step, it needs to satisfy $\lim_{n^\RN{1},n^\RN{2}\to\infty}\Pr\big[\widehat{ne}_j^0\supseteq ne_j^0\big]=1$ (Figure~\ref{fig:venn}b). On the other hand, if the cardinality of $\widehat{ne}_j^0$ grows large compared to that of $ne_j^0$, the power of the proposed framework deteriorates. In fact, if $\widehat{ne}_j^0 \supseteq ne_j^{\RN{1}}\cup ne_j^{\RN{2}}$, the differential connectivity of node $j$ cannot be detected. 
In the following subsections, we also discuss how the randomness in estimating $\widehat{ne}_j^0$ may affect the results of the hypothesis testing step and how valid inferences can be obtained. 

\def\firstcircle{(0,0) circle (0.8cm)}
\def\secondcircle{(0:0.5cm) circle (0.8cm)}
\def\thirdcircle{(0:0.25cm) circle (0.8cm)}
\def\fourthcircle{(0:0.25cm) circle (1.4 cm)}
\def\ellipse{(0:0.25cm) ellipse (1.2cm and 0.9cm)}

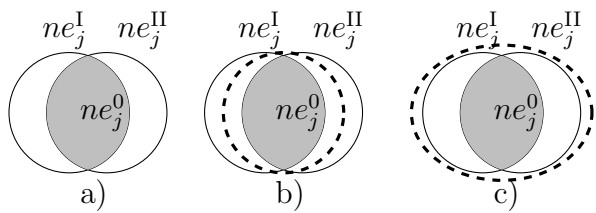
\begin{figure}[h]
\caption{Illustration of the common neighborhood $ne^0_j = ne_j^{\RN{1}}\cap ne_j^{\RN{2}}$ of node $j$ in two networks $\cE^{\RN{1}}$ and $\cE^{\RN{2}}$: In all figures, $ne^0_j$ is shaded in gray, and its estimate, $\hat{ne}^0_j$, is shown in dashed ovals; the unshaded parts of $ne_j^{\RN{1}}$ and $ne_j^{\RN{2}}$ correspond to $ne_j^{\RN{1}}\,\triangle\,ne_j^{\RN{2}}=\varnothing$. In (b), $\hat{ne}^0_j$ satisfies the \emph{coverage property} of Section~\ref{sec:est} and allows differential connectivity to be estimated; in (c), $ \hat{ne}_j^0 \supseteq ne_j^{\RN{1}}\cup ne_j^{\RN{2}}$ and differential connectivity of $j$ cannot be detected, as illustrated in Section~\ref{sec:lassopower}}\label{fig:venn}
\vspace{1em}
\centering
\begin{tikzpicture}
    \draw \firstcircle node[below] {};
    \draw \secondcircle node [above] {};

    \begin{scope}
      \clip \firstcircle;
      \fill[gray!50] \secondcircle;
    \end{scope}
    \node[text width=0.5cm, anchor=west, right] at (-0.5, 1.1) {$ne_j^{\RN{1}}$};
    \node[text width=0.5cm, anchor=west, right] at (0.5, 1.1) {$ne_j^{\RN{2}}$};
    \node[text width=0.5cm, anchor=west, right] at (0, 0) {$ne_j^0$};
    \node[text width=0.5cm, anchor=west, right] at (0, -1.1) {a)};

    \begin{scope}[shift={(2.6cm,0cm)}]
        \draw \firstcircle node {};
        \draw \secondcircle node {};
	\draw[dashed, very thick] \thirdcircle node [above] {};
	\begin{scope}
      	\clip \firstcircle;
      		\fill[gray!50] \secondcircle;
    	\end{scope}
    \end{scope}
    \node[text width=0.5cm, anchor=west, right] at (2.1, 1.1) {$ne_j^{\RN{1}}$};
    \node[text width=0.5cm, anchor=west, right] at (3.1, 1.1) {$ne_j^{\RN{2}}$};
    \node[text width=0.5cm, anchor=west, right] at (2.6, 0) {$ne_j^0$};
    \node[text width=0.5cm, anchor=west, right] at (2.6, -1.1) {b)};
   
    \begin{scope}[shift={(5.5cm,0cm)}]
        \draw \firstcircle node {};
        \draw \secondcircle node {};
	\begin{scope}
      	\clip \firstcircle;
      		\fill[gray!50] \secondcircle;
    	\end{scope}
	\draw[dashed, very thick] \ellipse node [above] {};
    \end{scope}
    \node[text width=0.5cm, anchor=west, right] at (5, 1.1) {$ne_j^{\RN{1}}$};
    \node[text width=0.5cm, anchor=west, right] at (6, 1.1) {$ne_j^{\RN{2}}$};
    \node[text width=0.5cm, anchor=west, right] at (5.5, 0) {$ne_j^0$};
    \node[text width=0.5cm, anchor=west, right] at (5.5, -1.1) {c)};

\end{tikzpicture}
\end{figure}

\subsection{Estimating Common Neighbors}\label{sec:est}

Given a $j\in\cV$, the first step of DCA involves obtaining an estimate $\widehat{ne}_j^0$ of $ne_j^0$. 
We do not need $\widehat{ne}_j^0$ to be a consistent estimate of $ne_j^0$, which usually requires stringent conditions \citep[see, e.g.,][]{MeinshausenBuhlmann2006, ZhaoYu2006}. Instead, we observe that under the null hypothesis $H_{0,j}^\ast:ne^{\RN{1}}_j=ne^{\RN{2}}_j$, we have $ne^{\RN{1}}_j=ne^{\RN{2}}_j=ne^{\RN{1}}_j\cup ne^{\RN{2}}_j=ne^{\RN{1}}_j\cap ne^{\RN{2}}_j\equiv ne_j^0$, which indicates that if $\widehat{ne}_j^0\supseteq ne_j^0$, then $\widehat{ne}_{j}^0\supseteq ne_j^{\RN{1}}\cup ne_j^{\RN{2}}$. 
In other words, if $\widehat{ne}_j^0\supseteq ne_j^0$, then under the null hypothesis $H_{0,j}:ne^{\RN{1}}_j= ne^{\RN{2}}_j$, there should be no $k\notin \widehat{ne}_{j}^0$ such that $k\in ne^{\RN{1}}_j\cup ne^{\RN{2}}_j$. Thus, we propose to test $H_{0,j}^\ast$ by examining whether there exists a $k\notin \widehat{ne}_{j}^0$ such that $k\in ne^{\RN{1}}_j\cup ne^{\RN{2}}_j$. 
Based on the above observation, we require that  
\begin{equation}\label{eqn:estevent}
	\lim_{n^\RN{1},n^\RN{2}\to\infty}\Pr[\widehat{ne}_j^0\supseteq ne_j^0]=1.
 \end{equation}
We call \eqref{eqn:estevent} the \emph{coverage property} of estimated common neighbors (see Figure~\ref{fig:venn}b).

Let $\bX^{\RN{1}}$ and $\bX^\RN{2}$ be two Gaussian datasets of size $n^{\RN{1}} \times p$ and $n^{\RN{2}} \times p$ containing measurements of the same set of variables $\cV$ (with $p=|\cV|$) in populations $\RN{1}$ and $\RN{2}$, respectively. 
Note that the data may be high-dimensional, i.e., $p \gg \max\{n^{\RN{1}}, n^{\RN{2}}\}$. To estimate the common neighborhood $ne_j^0$, for $m\in\{\RN{1}, \RN{2}\}$, we write 
\begin{align}
\bx_j^m&=\bX_{\backslash j}^m\bbeta^{m, j}+\bepsilon^{m,j}, \label{eq:GRep} 
\end{align}
where $\bbeta^{m,j}$ is a $(p-1)$-vector of coefficients and $\bepsilon^{m, j}$ is an $n^m$-vector of random errors. By Gaussianity, $\beta_k^{m, j}\neq0$ if and only if $\Omega_{jk}^m\neq0$, which, as discussed before, is equivalent to $k\in ne_j^m$. Therefore, the common neighbors of node $j$ in the two populations are
\begin{align}\label{eq:jointneighborbeta}
ne_j^0\equiv ne_j^{\RN{1}}\cap ne_j^{\RN{2}} = \left\{k:\beta_k^{\RN{1}, j}\neq0 \text{ \& } \beta_k^{\RN{2}, j}\neq0\right\}.
\end{align}

Based on \eqref{eq:jointneighborbeta}, an estimate of $ne_j^0$ may be obtained from the estimated supports of $\bbeta^{\RN{1}, j}$ and $\bbeta^{\RN{2}, j}$. 

Various procedures can be used to estimate $\bbeta^{\RN{1}, j}$ and $\bbeta^{\RN{2}, j}$ and, in turn, $ne_j^0$. We present a lasso-based estimator as an example in Section~\ref{sec:prop}. Proposition~\ref{thm:unbiased} shows that under appropriate conditions, the lasso-based estimate satisfies the coverage property \eqref{eqn:estevent}, and is thus valid for the estimation step of DCA.

\subsection{Testing Differential Connectivity}\label{sec:test}
Recall, from our discussion in the previous section, that the estimated joint neighborhood $\widehat{ne}_j^0$ needs to satisfy the coverage property $\lim_{n^\RN{1},n^\RN{2}\to\infty}\Pr[\widehat{ne}_j^0\supseteq ne_j^0]=1$. With $\widehat{ne}_j^0\supseteq ne_j^0$, if there exists a $k\notin \widehat{ne}_{j}^0$ such that $k\in ne^{\RN{1}}_j\cup ne^{\RN{2}}_j$, then with probability tending to one, $ne^{\RN{1}}_j\neq ne^{\RN{2}}_j$. 
As mentioned above, in GGMs, $k\in ne^{\RN{1}}_j\cup ne^{\RN{2}}_j$ if and only if $\beta_k^{\RN{1}, j}\neq 0$ or $\beta_k^{\RN{2}, j}\neq 0$. 
Thus, to determine whether there exists a $k\notin \widehat{ne}_{j}^0$ such that $k\in ne^{\RN{1}}_j\cup ne^{\RN{2}}_j$, we test the following hypotheses 
\begin{equation}
	H_{0, j}:\beta_k^{\RN{1}, j}= 0 \text{ \& } \beta_k^{\RN{2}, j}= 0,\quad\forall k\notin\widehat{ne}_{j}^0\cup\{j\},
\end{equation}
where $\widehat{ne}_j^0 = \{k:\hat\beta_k^{\RN{1}, j}\neq0 \text{ \& } \hat\beta_k^{\RN{2}, j}\neq0\}$.

Using the \v{S}id\'ak correction to control false positive rate of $H_{0, j}$ at level $\alpha>0$, we control false positive rates of $H_{0, j}^{\RN{1}}: \beta_k^{\RN{1}, j}= 0$ and $H_{0, j}^{\RN{2}}: \beta_k^{\RN{2}, j}= 0$ at the level $1-\sqrt{1-\alpha}$ for all $k\notin\widehat{ne}_{j}^0\cup\{j\}$. Note that if $\widehat{ne}_j^0\cup\{j\}=\cV$, we do not reject $H_{0,j}$. We will discuss later in this subsection how to test $H_{0, j}^{\RN{1}}$ and $H_{0, j}^{\RN{2}}$. 

The proposal outlined so far faces an important obstacle: Even when $\widehat{ne}_{j}^0$ satisfies the coverage property, the hypotheses $H_{0, j}^{\RN{1}}$ and $H_{0, j}^{\RN{2}}$ \emph{depend on the data} through $\widehat{ne}_{j}^0$, which is a random quantity. This dependence complicates hypothesis testing: under the current procedure, we are effectively looking at the same data twice, once to formulate hypotheses and once to test the formulated hypotheses. Conventional statistical wisdom suggests that this kind of double-peeking would render standard hypothesis testing procedures invalid \citep[see, e.g.,][]{LeebPotscher2008PoSI}. 

To overcome the above difficulty, we offer two different strategies. In the first, we apply sampling splitting to avoid looking at the data twice \citep[see, e.g.,][]{WassermanRoeder2009posi, Meinshausenetal2009posi}. In this approach, the data are divided in two parts; the first part is used to estimate $\widehat{ne}_j^0$ and the second to test  $H_{0, j}$. The second strategy is provided in Proposition~\ref{thm:consistent} in Section~\ref{sec:prop}, which shows that although $\widehat{ne}_{j}^0$ is in general random, under appropriate conditions, the lasso-based estimate of $\widehat{ne}_{j}^0$ discussed in Section~\ref{sec:prop} converges in probability to a \emph{deterministic} set, which is not affected by the randomness of the data. Thus, under those conditions, asymptotically, we can treat $\widehat{ne}_{j}^0$ as deterministic, and hence treat $H_{0, j}^{\RN{1}}$ and $H_{0, j}^{\RN{2}}$ as classical non-data-dependent hypotheses. 

To test $H_{0, j}^{\RN{1}}: \beta_k^{\RN{1}, j}= 0,\forall k\notin\widehat{ne}_{j}^0\cup\{j\}$ and $H_{0, j}^{\RN{2}}: \beta_k^{\RN{2}, j}= 0,\forall k\notin\widehat{ne}_{j}^0\cup\{j\}$, we can use recent proposals for testing coefficients in high-dimensional linear regression \citep[e.g.,][]{javanmard2013confidence, ZhangZhang2014LDPE, vandeGeeretal2014LDPE, ZhaoShojaie2015Grace, NingLiu2015decor}. 
To control false positive rates of $H_{0, j}^{\RN{1}}$ and $H_{0, j}^{\RN{2}}$, we need to control the family-wise error rate (FWER) on individual regression coefficients using, e.g., the Holm procedure \citep{Holm1979}. Alternatively, $H_{0, j}^{\RN{1}}$ and $H_{0, j}^{\RN{2}}$ can be tested using group hypothesis testing procedures that examine a group of regression coefficients, such as the least-squares kernel machines (LSKM) test \citep{Liuetal2007}. Although such group hypothesis testing approaches cannot be used to infer which specific edges show differential connectivity, they often result in advantages in computation and statistical power for testing $H^\ast_{0, j}$ compared to hypothesis testing approaches that examine individual regression coefficients. 

Because edges with different dependency relationship in two networks must also have different strength of connectivity, in practice, we can first apply methods described in Section~\ref{sec:challenges} to find edges that show different strengths of connectivity in two networks. Then, restricted to edges that are found to have different connectivity strength, we can apply DCA to find edges that are differentially connected. Such a procedure may deliver improved power and false positive rate in ultra-high dimensional settings.
Finally, we conclude that two networks are differentially connected if any of the node-wise hypothesis $H^\ast_{0,j}: ne_j^{\RN{1}}= ne_j^{\RN{2}}$ is rejected. Thus, to control the network-wise false positive rate of testing $G^\RN{1}=G^\RN{2}$, we should control the family-wise error rate for node-wise tests using, e.g., the Holm procedure \citep{Holm1979}.

To summarize, DCA consists of two steps: estimation and hypothesis testing. These steps do not require specific methods. For the estimation step, we require that for each $j\in\cV$, the estimated common neighborhood, $\widehat{ne}_j^0$, satisfies the coverage property $\lim_{n^\RN{1},n^\RN{2}\to\infty}\Pr[\widehat{ne}_j^0\supseteq ne_j^0]=1$. Moreover, we require that either $\widehat{ne}_j^0$ is deterministic with high probability through, e.g., Proposition~\ref{thm:consistent}, or that the dependence between the estimation and hypothesis testing steps is severed by sample-splitting. 
For the hypothesis testing step, any valid high-dimensional hypothesis testing method that examines individual regression coefficients or a group of them is suitable. We arrive at the following theorem. 
\begin{theorem}\label{thm:DCA}
Suppose the procedure used in the estimation step of DCA satisfies the following conditions for each $j \in \cV$:  
\begin{enumerate}
\item The estimated common neighborhood of node $j$, $\widehat{ne}_j^0$, satisfies the coverage property, i.e., $\lim_{n^\RN{1},n^\RN{2}\to\infty}\Pr[\widehat{ne}_j^0\supseteq ne_j^0]=1$;
\item Either the estimated common neighborhood $\widehat{ne}^0_j$ is deterministic with probability tending to one, or the data used to test hypotheses $H_{0, j}^{m}$ for $m\in\{\RN{1}, \RN{2}\}$ are independent of the data used to estimate $\widehat{ne}^0_j$.
\end{enumerate}
Then, if for $m\in\{\RN{1}, \RN{2}\}$ the hypothesis testing procedure for testing $H_{0, j}^{m}:\bbeta^{m,j}_{\backslash\widehat{ne}_j^0}=\bzero$ is asymptotically valid, DCA asymptotically controls the false positive rate of $H^\ast_{0,j}:ne^{\RN{1}}_j=ne^{\RN{2}}_j$.
\end{theorem}
Theorem~\ref{thm:DCA} outlines a general framework that can incorporate many estimation and inference procedures. In particular, any method that provides a \emph{consistent} estimate of $ne^0_j$ asymptotically satisfies the conditions of the theorem, because with high probability, $\widehat{ne}_j^0 = ne_j^0$, which is, obviously, a deterministic set. However, consistent variable selection in high dimensions often requires stringent assumptions that may not be justified. In Section~\ref{sec:prop}, we discuss two alternative strategies based on lasso that satisfy the requirements of Theorem~\ref{thm:DCA} under milder assumptions. 

\subsection{The Validity of Lasso for DCA}\label{sec:prop}
A convenient procedure for estimating $ne_j^0$ is the \emph{lasso neighborhood selection} \citep{MeinshausenBuhlmann2006}. In this section, we show that lasso is a valid procedure for the estimation step in DCA. However, it is not the only valid estimation procedure: any procedure that satisfies the requirements of Theorem~\ref{thm:DCA} is valid. We discuss the power of DCA with lasso as the estimation procedure in Section~\ref{sec:lassopower}. There, we also present a sufficient condition for the DCA to asymptotically achieve perfect power.

In this section, we present two propositions regarding lasso neighborhood selection, which show that under appropriate conditions, the estimate $\widehat{ne}_j^0$ satisfies the coverage property, $\lim_{n^\RN{1},n^\RN{2}\to\infty}\Pr[\widehat{ne}_j^0\supseteq ne_j^0]=1$, and is deterministic with high probability. 
Together, these results imply that lasso neighborhood selection satisfies the requirements of Theorem~\ref{thm:DCA}. 

Note that, even if the estimated common neighborhood $\widehat{ne}^0_j$ is not deterministic with high probability, we can still apply sample splitting to obtain a valid estimation procedure for DCA based on lasso that satisfies the requirements of Theorem~\ref{thm:DCA}. The validity of the lasso neighborhood selection with sample splitting is established in \cite{WassermanRoeder2009posi, Meinshausenetal2009posi}. 

To establish that lasso-based estimates of neighborhoods are deterministic with high probability, in Proposition~\ref{thm:consistent} we establish a novel relationship between the lasso neighborhood selection estimator, 
\begin{align}\label{eq:lasso}
\hat\bbeta^{m, j}\equiv\argmin_{\bb \in \mathbb{R}^{p-1}}\left\{\frac{1}{2n}\left\|\bx^m_j-\bX^m_{\backslash j} \bb\right\|_2^2 + \lambda_j^m\left\|\bb\right\|_1 \right\}.
\end{align} 
and its \emph{noiseless} (and hence deterministic) counterpart
\begin{align}\label{eq:nllasso}
	\tilde\bbeta^{m, j}\equiv\argmin_{\bb \in \mathbb{R}^{p-1}}\left\{\E\left[\frac{1}{2n}\left\|\bx^m_j-\bX^m_{\backslash j} \bb\right\|_2^2\right] + \lambda_j^m\left\|\bb\right\|_1 \right\}.
\end{align}

We now present Propositions~\ref{thm:unbiased} and ~\ref{thm:consistent}. As mentioned in Section~\ref{sec:summary}, Proposition~\ref{thm:unbiased} implies that lasso neighborhood selection is a valid method for estimation in our framework, and Proposition~\ref{thm:consistent} relieves us from using sample-splitting to circumvent double-peeking by our procedure. The conditions are summarized in Section~\ref{sec:conditions}. Note that we only present lasso here as an example---other methods can be incorporated into DCA so long as they satisfy the requirements of Theorem~\ref{thm:DCA}. A number of methods that fit into the DCA framework will be numerically evaluated in Section~\ref{sec:sims}.
\begin{proposition}\label{thm:unbiased}
Suppose conditions ({\bf A1}) and ({\bf A2}) in Section~\ref{sec:conditions} hold for variable $j\in\cV$. Then $\widehat{ne}_j^0$ estimated using lasso neighborhood selection satisfies
\begin{align}
\lim_{n^\RN{1},n^\RN{2}\to\infty}\Pr\left[\widehat{ne}_{j}^0\supseteq ne_j^0\right]&=1.
\end{align} 
\end{proposition}

\begin{proposition}\label{thm:consistent}
Suppose conditions ({\bf A1}) -- ({\bf A3}) in Section~\ref{sec:conditions} hold for variable $j\in\cV$. 
Then $\widehat{ne}_j^0$ estimated using lasso neighborhood selection satisfies 
\begin{align}
\lim_{n^\RN{1},n^\RN{2}\to\infty}\Pr\left[\widehat{ne}_{j}^0= \widetilde{ne}_{j}^0\right]&= 1, 
\end{align}
where $\widetilde{ne}_{j}^0\equiv\supp(\tilde\bbeta^{\RN{1}, j})\cap\supp(\tilde\bbeta^{\RN{2}, j})$, and $\tilde\bbeta^{\RN{1}, j}$ and $\tilde\bbeta^{\RN{2}, j}$ are defined in \eqref{eq:nllasso}.
\end{proposition}

Propositions~\ref{thm:unbiased} and \ref{thm:consistent} are proved in Sections~\ref{sec:pfunbiased} and \ref{sec:pfconsistent}, respectively. The result in Proposition~\ref{thm:consistent} should not be confused with the variable selection consistency of lasso \citep{MeinshausenBuhlmann2006}, which shows that under the stringent irrepresentability condition, the selected neighborhoods converge in probability to the true neighborhoods, i.e., $\lim_{n^\RN{1},n^\RN{2}\to\infty}\Pr\big[\widehat{ne}_{j}^0= ne_j^0\big]= 1$. 
Proposition~\ref{thm:consistent} only shows that the selected neighborhoods converge to \emph{deterministic sets}, $\widetilde{ne}_j^0$.

\subsection{Power of DCA with Lasso in the Estimation Step}\label{sec:lassopower}
In Section~\ref{sec:est}, we argued that the estimated common neighborhood $\widehat{ne}^0_j$ needs to satisfy the coverage property, i.e., $\lim_{n^\RN{1},n^\RN{2}\to\infty}\Pr[\widehat{ne}_j^0\supseteq ne_j^0]=1$. 
In this section, we discuss how the cardinality of $\widehat{ne}^0_j$ affects the power of DCA. 
We also discuss a sufficient condition, where, using lasso in the estimation step, the power of DCA could approach one asymptotically for detecting differential connectivity. 

As mentioned in Section~\ref{sec:test}, to examine $H^\ast_{0,j}: ne_j^\RN{1}\neq ne_j^\RN{2}$, in the second step of DCA, we test whether variable $j$ is conditionally independent of variables that are not in the estimated common neighborhood. 
In the case where $ne_j^\RN{1}\neq ne_j^\RN{2}$, if the estimated neighborhood of variable $j$ is too large, such that $\widehat{ne}_j^0\supseteq ne_{j}^\RN{1} \cup ne_j^\RN{2}$, then for any $k\notin\widehat{ne}_j^0\cup\{j\}$, $\bx_j^m\independent \bx^m_k\, \mid \,\bx^m_{\backslash\{j,k\}}$ for $m\in\{\RN{1},\RN{2}\}$. In this case, we will not be able to identify differential connectivity of node $j$. Thus, even though the validity of DCA requires that $\widehat{ne}_j^0$ achieves the coverage property, $\widehat{ne}_j^0$ should not be exceedingly larger than $ne_j^0$. 

To examine the power of DCA with the lasso-based estimate of $\widehat{ne}_j^0$, suppose $|ne_j^\RN{2}|=\smallO(|ne_j^\RN{1}|)$. 
\citet{BelloniChernozhukov2013qlambda} show that, under mild conditions, with high probability, $|ne_j^m|\asymp|\widehat{ne}_j^m|=|\supp\big(\hat\bbeta^{m,j}\big)|$, where $\asymp$ denotes that two quantities are of the same asymptotic order. Therefore, with high probability, $|\widehat{ne}_j^0| \le |\widehat{ne}_j^\RN{2}|\asymp |ne_j^\RN{2}|=\smallO(|ne_j^\RN{1}|)$, i.e., $|ne_j^\RN{1}|\gg |\widehat{ne}_j^0|$ so that $\widehat{ne}_j^0\nsupseteq ne_{j}^\RN{1}$. Similarly, if $|ne_j^\RN{1}|=\smallO(|ne_j^\RN{2}|)$, then with high probability $\widehat{ne}_j^0\nsupseteq ne_{j}^\RN{2}$. Thus, if $|ne_j^\RN{1}|$ and $|ne_j^\RN{2}|$ are not of the same order, then, with high probability, $\widehat{ne}_j^0\nsupseteq ne_{j}^\RN{1}\,\triangle\,ne_j^\RN{2}$, and there exists $k\notin\widehat{ne}_j^0\cup\{j\}$ such that $\bx_j^m\nindependent \bx^m_k\,|\,\bx^m_{\backslash\{j,k\}}$ for $m\in\{\RN{1},\RN{2}\}$. In this case, with any conditional testing method that achieves asymptotic power one, DCA is asymptotically guaranteed to detect the differential connectivity of $j$. 

While the conditions presented in the above special case are sufficient and not necessary, the scenario sheds light on the power properties of DCA. We defer to future research a more thorough assessment of power properties of DCA.

\section{Numerical Studies}
\subsection{Simulation Studies}\label{sec:sims}

In this section, we present results of a simulation study that evaluates the power and false positive rate of the DCA framework using various choices of procedures in the estimation and hypothesis testing steps. As discussed in Section~\ref{sec:challenges}, quantitative tests that examine the equality of partial correlations do not control the type-I error rate of qualitative tests. Therefore, comparison with these methods would not be meaningful.

In this simulation, we generate $\cE^{\RN{1}}$ from a power-law degree distribution with power parameter 5, $|\cV|\equiv p=200$ and $|\cE^{\RN{1}}|=p(p-1)/100$; this corresponds to an edge density of 0.02 in graph $G^{\RN{1}}$. Power-law degree distributions are able to produce graphs with hubs, which are expected in real-world networks \citep{Newman2013}. To simulate $\cE^{\RN{2}}$, among the 100 most connected nodes in $G^{\RN{1}}$, we randomly select 20 nodes, remove all the edges that are connected to them, and then randomly add edges to graph $G^{\RN{2}}$ so that $|\cE^{\RN{2}}|=|\cE^{\RN{1}}|$. To simulate $\bOmega^{\RN{1}}$, for $j\neq k$, we let 
\[   \Omega_{jk}^{\RN{1}}=\left\{
\begin{array}{ll}
     0 		& (j,k)\notin\cE^{\RN{1}} \\
     0.5	& (j, k)\in\cE^{\RN{1}}, \text{ with 50\% probability} \\
    -0.5	& (j,k)\in\cE^{\RN{1}}, \text{ with 50\% probability} \\
\end{array} 
\right.. \]
To simulate $\bOmega^{\RN{2}}$, for $j\neq k$, we let 
\[   \Omega_{jk}^{\RN{2}}=\left\{
\begin{array}{ll}
	\Omega_{jk}^{\RN{1}} & (j,k)\in\cE^{\RN{1}}\cap\cE^{\RN{2}} \\
     0 		& (j,k)\notin\cE^{\RN{2}} \\
     0.5	& (j, k)\in\cE^{\RN{2}}\backslash\cE^{\RN{1}}, \text{ with 50\% probability} \\
    -0.5	& (j,k)\in\cE^{\RN{2}}\backslash\cE^{\RN{1}}, \text{ with 50\% probability} \\
\end{array} 
\right.. \]
Finally, for $m\in\{\RN{1},\RN{2}\}$, we let $\Omega_{jj}^m=\sum_{k\neq j}\big|\Omega_{jk}^m\big|+u^m$ for $j=1,\dots,p$, where $u^m$ is chosen such that $\phi^2_{\min}\big(\bOmega^m\big)=0.1$, where $\phi^2_{\min}\big(\bOmega^m\big)$ is the smallest eigenvalues of $\bOmega^m$. Figure~\ref{fig:hist} shows the distribution of non-zero partial correlations in $\bOmega^{\RN{1}}$ and $\bOmega^{\RN{2}}$. From $\bOmega^{\RN{1}}$ and $\bOmega^{\RN{2}}$, we generate $\bX^{\RN{1}}\sim_{i.i.d.}\mathcal{N}_p\big(\bzero,\big[\bOmega^{\RN{1}}\big]^{-1}\big)$ and $\bX^{\RN{2}}\sim_{i.i.d.}\mathcal{N}_p\big(\bzero,\big[\bOmega^{\RN{2}}\big]^{-1}\big)$, where $n^{\RN{1}}=n^{\RN{2}}=n \in\{ 100, 200, 400, 800\}$.

\begin{figure}[h]
\caption{Distribution of non-zero partial correlations in simulated $\bOmega^{\RN{1}}$ and $\bOmega^{\RN{2}}$. 
} 
\vspace{1em}
\label{fig:hist}
\centering
\includegraphics[width = 0.8\textwidth]{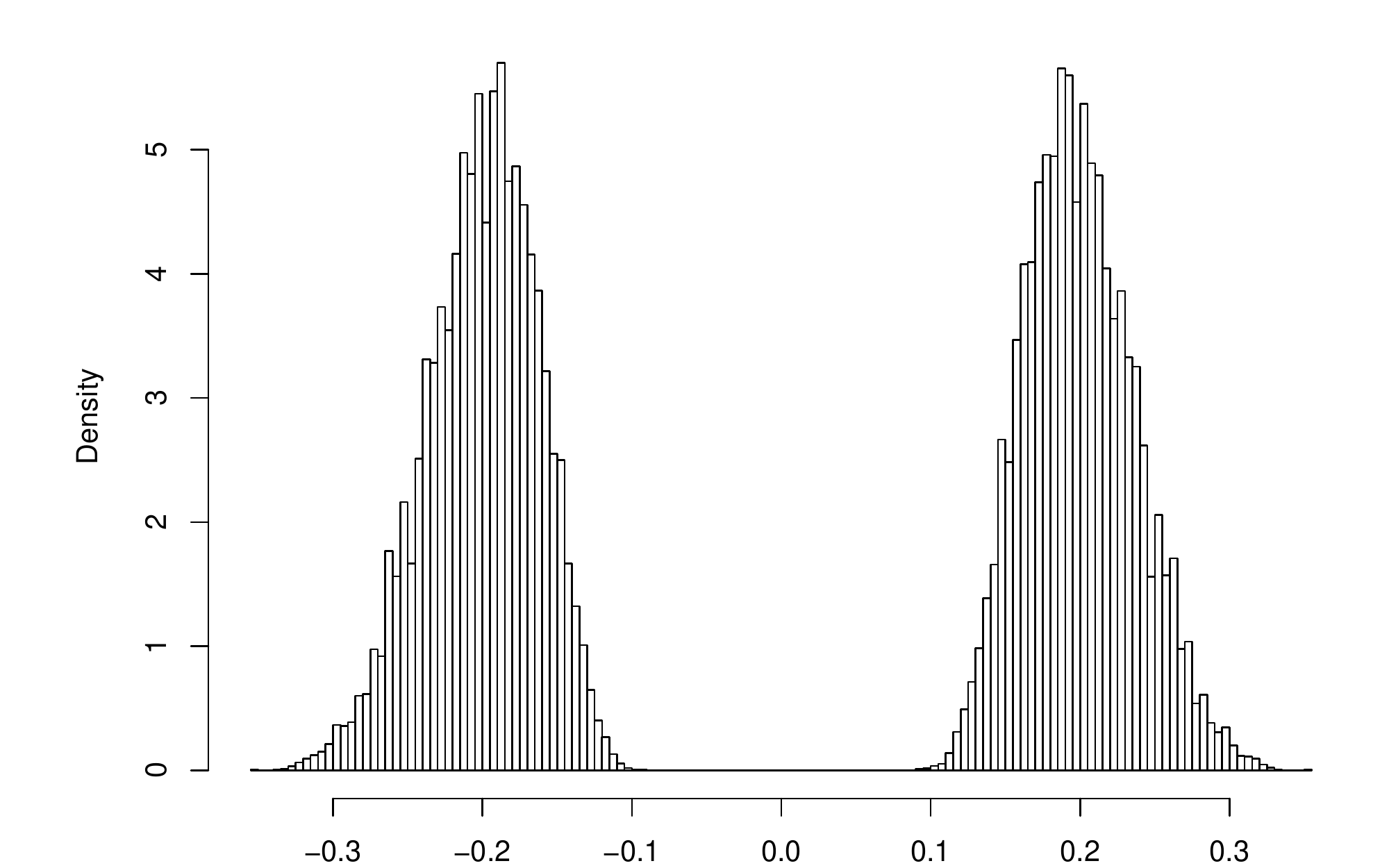}
\end{figure}


To estimate common neighbors of each node $j\in\cV$, we use lasso neighborhood selection, with tuning parameters chosen by 10-fold cross-validation (CV). 
We either use sample-splitting to address the issue of double-peeking, with half of samples used to estimate $\widehat{ne}_j^0$ and the other half to test $H_{0, j}$, or use a na\"ive approach, where the whole dataset is used to estimate $\widehat{ne}_j^0$ and to test $H^\ast_{0, j}: ne^{\RN{1}}_j=ne^{\RN{2}}_j$; the latter approach is justified by Proposition~\ref{thm:consistent}. 
To examine $H_{0, j}$ for each $j=1,\dots,p$, we consider LSKM \citep{Liuetal2007}, which is a group hypothesis testing methods, and the GraceI test \citep{ZhaoShojaie2015Grace}, which examines individual regression coefficients. As a result, we compare in total 4 approaches: \{na\"ive lasso neighborhood selection, sample-splitting lasso neighborhood selection\}$\times$\{LSKM, GraceI\}. 

Note that similar to the discussion in \citet{MeinshausenBuhlmann2006}, it is possible that a node-pair $(j, k)$ is identified to be differentially connected in testing  $H^\ast_{0, j}: ne^{\RN{1}}_j=ne^{\RN{2}}_j$, but not so in testing  $H^\ast_{0, k}: ne^{\RN{1}}_k=ne^{\RN{2}}_k$. To mitigate this issue, we used the ``OR'' rule \citep{MeinshausenBuhlmann2006} in the simulation studies and the cancer genetics application presented in Section~\ref{sec:RD}. With the ``OR'' rule, an edge becomes a false positive if it is a false positive in any of the two node-wise tests. Hence, we should control the false positive rate at level $\alpha/2$ for the node-wise tests (i.e., Bonferroni correction). Based on a similar reasoning, the ``AND'' rule is also valid if the node-wise tests are controlled at level $\alpha$.

Figure~\ref{fig:simGraceI} shows average false positive rates of falsely rejecting $H^\ast_{0, j}: ne^{\RN{1}}_j=ne^{\RN{2}}_j$, as well as average power of various DCA variants based on $R=100$ repetitions. 
Let $z_{j,r}$ be the decision function based on the GraceI test or LSKM: specifically, $z_{j,r}=1$ if hypothesis $H^\ast_{0, j}: ne^{\RN{1}}_j=ne^{\RN{2}}_j$ is rejected in the $r$th repetition, and $z_{j,r}=0$ otherwise. The average false positive rate is defined as
\begin{equation}
\mathrm{T1ER}= \frac{ \sum_{r=1}^R{\left\{\sum_{j\in\cV:ne^{\RN{1}}_{j,r}=ne^{\RN{2}}_{j,r}}z_{j,r}\right\}} }{ {\sum_{r=1}^R \left|\left\{j\in\cV:ne^{\RN{1}}_{j,r}=ne^{\RN{2}}_{j,r}\right\}\right|} },
\end{equation}
i.e., the proportion of null hypotheses in $R$ repetitions that we falsely reject $H^\ast_{0, j}$. 
For $t \in \{1,3,5,10\}$, the average power of rejecting $H^\ast_{0, j}: ne^{\RN{1}}_j=ne^{\RN{2}}_j$ when $ne^{\RN{1}}_j$ and $ne^{\RN{2}}_j$ differ by at least $t$ members is defined as 
\begin{equation}
\mathrm{Pt}= \frac{ \sum_{r=1}^R{\left\{\sum_{j\in\cV:\left|ne^{\RN{1}}_{j,r}\,\triangle\,ne^{\RN{2}}_{j,r}\right|\geq t}z_{j,r}\right\}} }{ 
\sum_{r=1}^R \left|\left\{j\in\cV:\left|ne^{\RN{1}}_{j,r}\,\triangle\,ne^{\RN{2}}_{j,r}\right|\geq t\right\}\right|
},
\end{equation}
where ``$\triangle$'' denotes the symmetric difference of two sets.

\begin{figure}[t]
\caption{The average false positive rate and power of rejecting $H^\ast_{0, j}$. The axis for false positive rate is on the left of each panel, whereas the axis for power is on the right.
} 
\vspace{1em}
\label{fig:simGraceI}
\centering
\includegraphics[width = 0.85\textwidth]{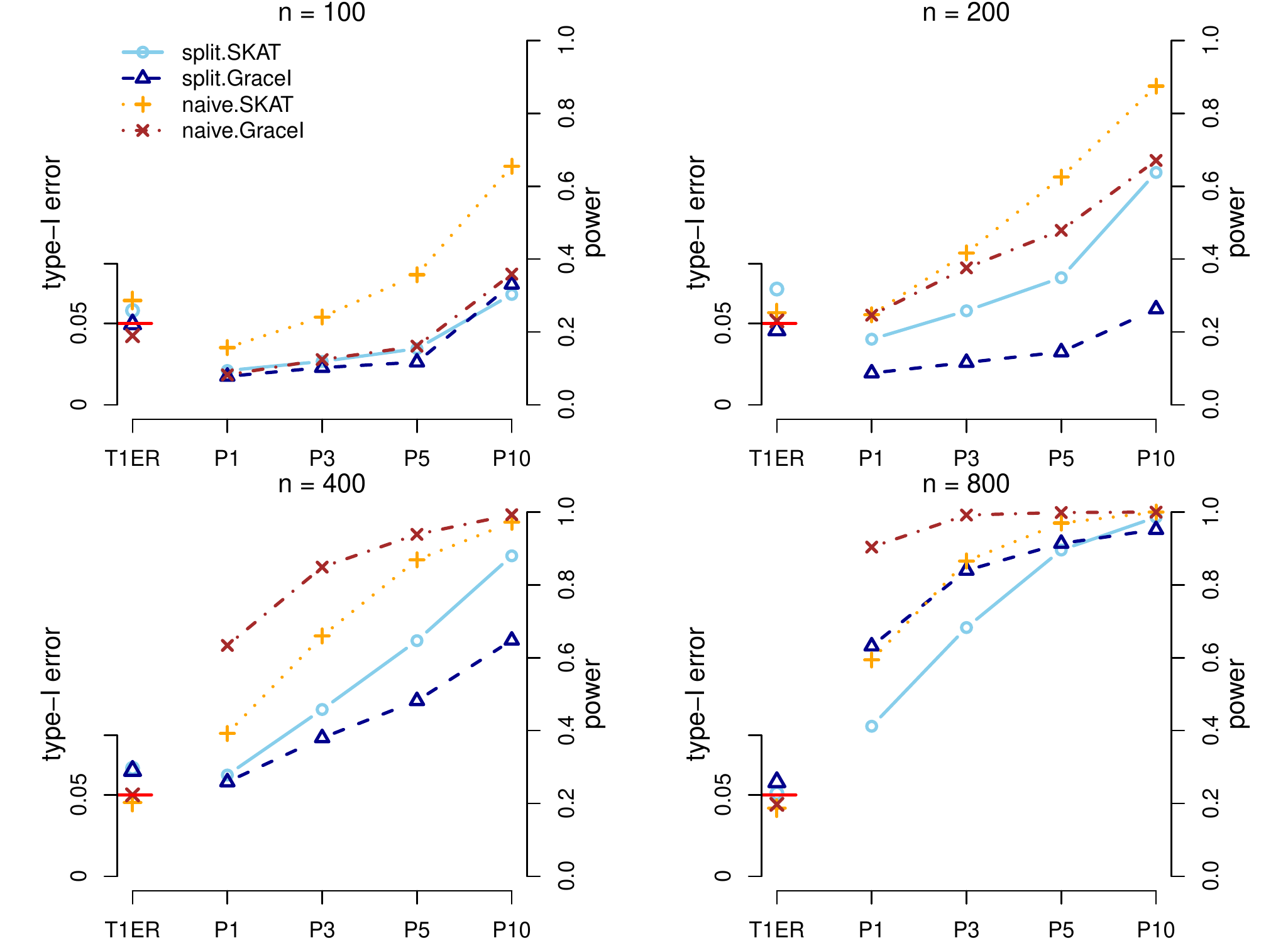}
\end{figure}

The simulation reveals several interesting patterns. First, na\"ive procedures which use the same data to estimate $\widehat{ne}_j^0$ and test $H^m_{0, j}, m \in \{ \RN{1}, \RN{2} \}$ tend to have better statistical power than their sample-splitting counterparts. This is understandable, as sample-splitting only uses half of the data for hypothesis testing. More surprisingly, na\"ive procedures also better control the false positive rate than sample-splitting procedures. This is because the event $\widehat{ne}_j^0\supseteq ne^0_j$, which is crucial for controlling the false positive rate and is guaranteed to happen with high probability asymptotically, is less likely to happen with the smaller samples available for the sample-splitting estimator.  
In addition, we can see that LSKM has better power than the GraceI test for smaller sample sizes (LSKM also has a slightly worse control of the false positive rate than GraceI). But as sample size increases, the power of GraceI eventually surpasses LSKM. Finally, as expected, the probability of rejecting $H^\ast_{0, j}: ne^{\RN{1}}_j=ne^{\RN{2}}_j$ is higher when $ne^{\RN{1}}_j$ and $ne^{\RN{2}}_j$ differ by more elements. 

%

\begin{figure}[t]
\caption{Differentially connected edges between ER- and ER+ breast cancer patients. Yellow edges are genetic interactions that are found in ER- but not ER+ breast cancer patients by the GraceI test; gray edges are genetic interactions that are found in ER+ but not ER- breast cancer patients by the GraceI test. Identically connected edges are omitted.
} 
\vspace{1em}
\label{fig:TCGAres}
\centering
\includegraphics[width = 0.75\textwidth, clip=TRUE, trim=1cm 2.5cm 1cm 2.5cm]{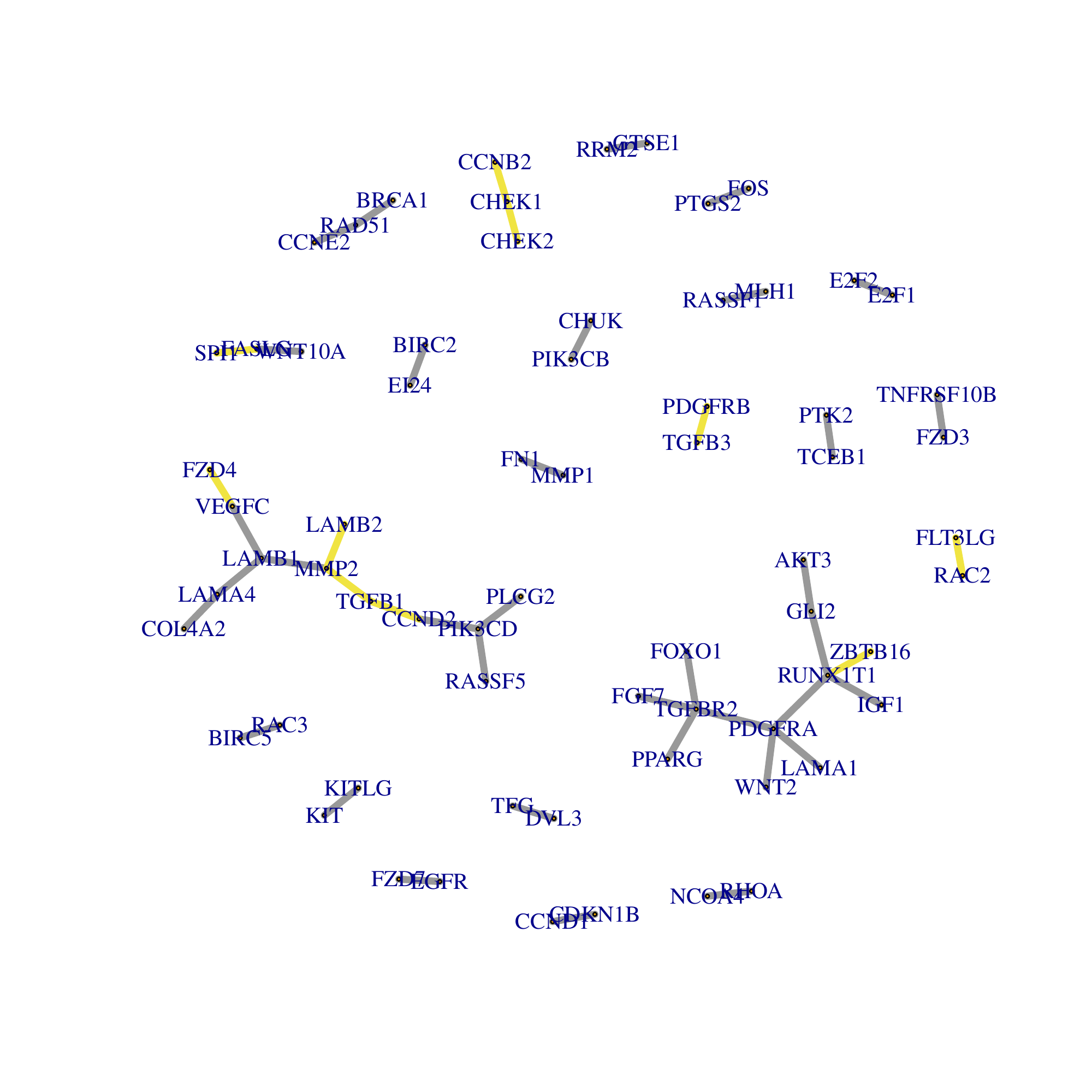}
\end{figure}

\subsection{Application to Cancer Genetics Data}\label{sec:RD}
Breast cancer has multiple clinically verified subtypes \citep{Perouetal2000} that have been shown to have distinct prognostics \citep{Jonssonetal2010}. Based on the expression of estrogen receptor (ER), breast cancer can be classified into ER positive (ER+) and ER negative (ER-) subtypes. ER+ breast cancer has a larger number of estrogen receptors, and has better survival prognosis than ER- breast cancer \citep{Careyetal2006ER}. The genetic pathways of ER+ and ER- subtypes are expected to be similar, but also show some important differences. Understanding such differences could be critical to help researchers better understand breast cancer. To investigate differences in genetic pathways between ER+ and ER- breast cancer patients, we obtain gene expression data from the Cancer Genome Atlas (TCGA). The data contain the expression levels of $p=358$ genes in cancer related pathways from KEGG for $n^{\RN{1}}=117$ ER- and $n^{\RN{2}}=407$ ER+ breast cancer patients.

Since our goal is to identify differentially connected edges, group hypothesis testing procedures such as LSKM are no longer valid. Hence, we use the GraceI test after na\"ive lasso neighborhood selection to examine the difference in genetic pathways between ER+ and ER- breast cancer patients. In this example, family-wise error rate is controlled at $\alpha = 0.1$ level using the Holm procedure. Differentially connected edges are shown in Figure~\ref{fig:TCGAres}. Specifically, among other genes, all of the genes that have at least three differential connections identified by DCA have already been found by previous research to be associated with the subtype, progression and prognostics of breast cancer. These \emph{highly differentially connected genes} are: laminin subunit $\beta$1 (LAMB1) \cite{Pellegrinietal1995}, matrix Metalloproteinase-2 (MMP2) \citep{JezierskaMotyl2009}, platelet-derived growth factor receptor $\alpha$ (PDGFRA) \cite{Carvalhoetal2005}, phosphoinositide 3'-kinases $\delta$ (PIK3CD) \citep{Sawyeretal2003}, runt-related transcription factor 1 (RUNX1T1) \citep{Janes2011} and TGF-$\beta$ receptor type-2 (TGFBR2) \citep{Maetal2012, Buschetal2015}.

As a comparison, we performed quantitative test of \citet{Gilletal2014}, which detected that 196 out of 358 genes in the dataset are differentially connected in two networks at family-wise error rate level of $\alpha = 0.1$. Given the overall robustness of biological system \citep[see, e.g.,][]{Kitano2004}, such a large number of differentially connected genes likely confirms our simulation findings that quantitative test does not control the null hypothesis $H^\ast_{0,j}: ne_j^{\RN{1}}= ne_j^{\RN{2}}$ at the desired level.

\subsection{Application to Brain Imaging Data}\label{sec:RD2}
Mild, uncomplicated traumatic brain injury (TBI), or concussion, can occur from a variety of head injury exposures. In youth, sports and recreational activities comprise a predominate number of these exposures with uncomplicated mild comprising the vast majority of TBIs. By definition, these are diagnostically ambiguous injuries with no radiographic findings on conventional CT or MRI.  While some children recover just fine, a subset remain symptomatic for a sustained period of time. This group---often referred to as the `miserable minority'---make up the majority of the patient population in concussion clinics.  Newer imaging methods are needed to provide more sensitive diagnostic and prognostic screening tools that elucidate the underlying pathophysiological changes in these concussed youth whose symptoms do not resolve.  To this end, a collaborative team evaluated 10-14 year olds following a sports or recreational concussion who remained symptomatic at 3-4 weeks post-injury and a group of age and gender matched controls with no history of head injury, psychological health diagnoses, or learning disabilities.  Advanced neuroimaging was collected on each participant which included collecting diffusion tensor imaging (DTI).  DTI has been shown to be sensitive to more subtle changes in white matter that have been reported to strongly correlate with axonal injury pathology \citep{bennett2012, mac2007} and relate to outcome in other concussion groups \citep{bazarian2012, cubon2011, gajawelli2013}.  

Upon preprocessing the data according to the procedure outlined in Section~\ref{sec:data}, we obtained data on $p=78$ brain regions from $n^{\RN{1}} = 27$ healthy controls and $n^{\RN{2}} = 25$ TBI patients. To assess whether brain connectivity patterns of TBI patients differs from that of healthy controls, we used the DCA framework with the GraceI test after na\"ive lasso neighborhood selection. We chose the lasso tuning parameter using10-fold CV and controlled the FWER of falsely rejecting $H^\ast_{0, j}: ne^{\RN{1}}_j=ne^{\RN{2}}_j$ for any $j=1,\dots, p$ at level 0.1 using the Holm procedure. The resulting brain connectivity networks are shown in Figure~\ref{fig:diffedge}, where differentially connected and common edges are drawn in different colors. It can be seen that a number of connections differ between TBI patients and healthy controls. The DTI data used in this study provide characterize microstructural changes in brain regions and not their functions. The assumption of multivariate normality may also not be realistic in this application. Finally, the study is based on small samples. Nonetheless, the results suggest orchestrated structural changes in TBI patients that may help form new hypotheses.

\begin{figure}
\caption{Common and differentially connected edges between concussed youth athletes and matched controls. Red and blue nodes are brain regions on the left and right cerebral cortices, respectively, whereas pink and turquoise nodes are other regions in the left and right brains. Gray edges are estimated common brain connections based on lasso neighborhood selection; blue edges are connections that are found in healthy controls but not in TBI patients;  
red edges are connections that are found in TBI patients but not in healthy controls. 
} 
\vspace{1em}
\label{fig:diffedge}
\centering
\includegraphics[angle=90, width = 9cm]{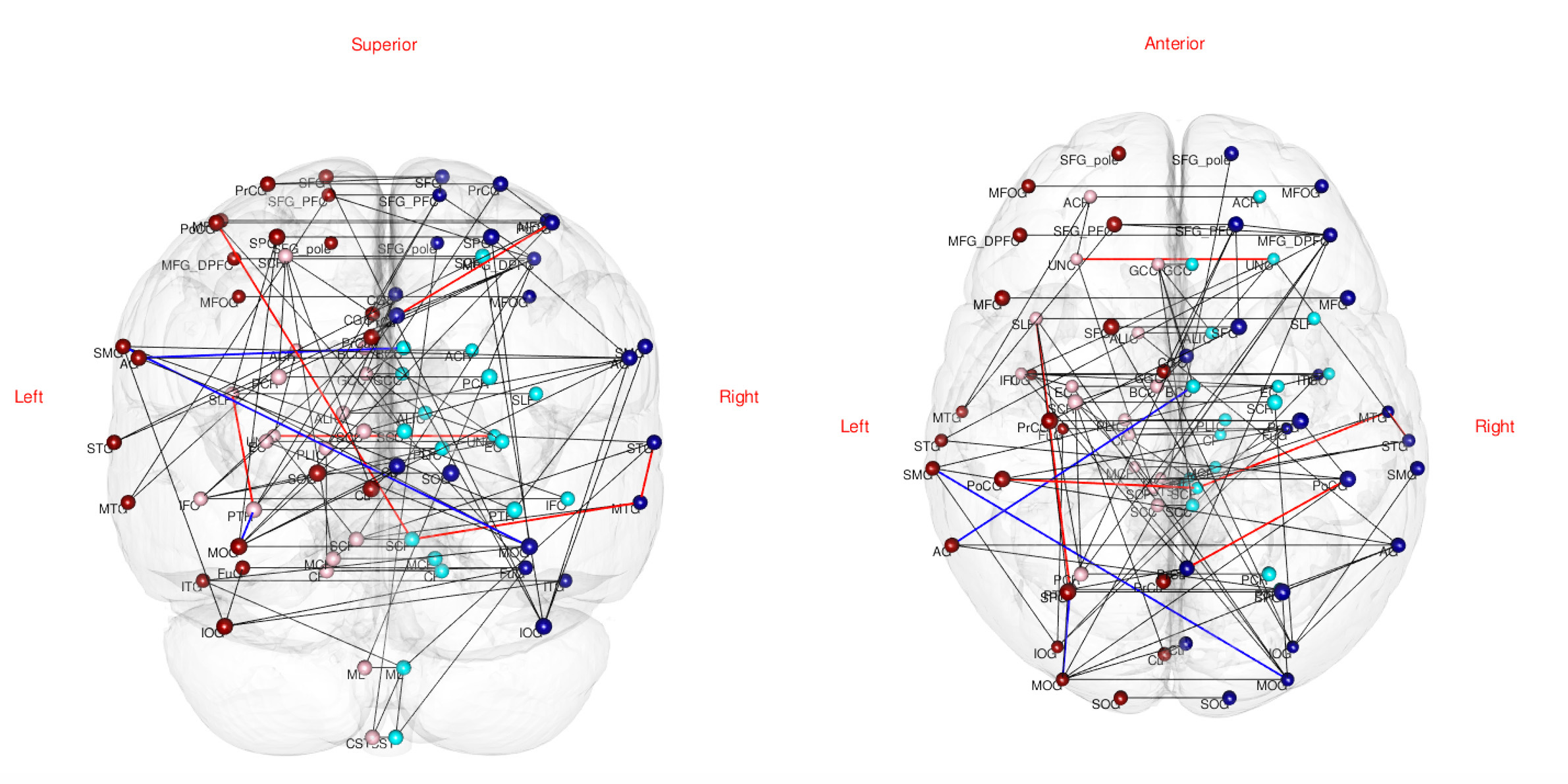}
\end{figure}

\section{Conclusion}\label{sec:disc}
In this paper, we highlighted challenges of identifying differential connectivity in high-dimensional networks using existing approaches, and proposed a new hypothesis testing framework, called \emph{differential connectivity analysis} (DCA), for identifying differences in two networks. 

DCA can incorporate various estimation and hypothesis testing methods, and can be easily extended to test for differential connectivity in multiple networks. 
Here, we considered two methods for estimation and inference: sample-splitting, which breaks down the dependence between estimation and hypothesis testing, and na\"ive inference, which utilizes the fact that the estimated support of lasso is deterministic with high probability. Besides sample splitting and na\"ive inference, another option is to build on recent advances in conditional hypothesis procedures \citep[see, e.g.,][]{lee2015exact,Tibshiranietal2016PoSI}. We leave to future work the exploration of whether conditional hypothesis testing procedures can be adapted to fit into DCA. Exploring the feasibility of incorporating non-convex estimation methods \citep[e.g.,][]{FanLi2001SCAD, Zhang2010MCP} in DCA could also be a fruitful area of research.

\section{Conditions for the Validity of Lasso for DCA} \label{sec:conditions}

The following are sufficient conditions for our propositions. 
\begin{itemize}
\item[({\bf A1})] For $m\in\{\RN{1},\RN{2}\}$, rows of the data $\bX^m$ are independent and identically distributed Gaussian random vectors: $\bX^m\sim_{i.i.d.}\mathcal{N}_p(\bzero, \bSigma^m)$, where, without loss of generality, we assume $\diag(\bSigma^m)=\bone$. Further, the minimum and maximum eigenvalues of $\bSigma^m$ satisfy
\[
\liminf_{n^m\to\infty}\phi^2_{\min}\left(\bSigma^m\right)>0\hspace{1em}\text{and}\hspace{1em}\limsup_{n^m\to\infty}\phi^2_{\max}\left(\bSigma^m\right)<\infty.
\]
\item[({\bf A2})] For $m\in\{\RN{1},\RN{2}\}$ and a given variable $j\in\cV$, the sample size $n^m$, dimension $p$, lasso neighborhood selection tuning parameters $\lambda_j^m$, number of neighbors $q^m_j\equiv|ne_j^m|$, and minimum non-zero coefficients $b_{\min}^{m,j} \equiv \min\{|\beta^{m,j}_{k}|: \beta^{m,j}_{k}\neq0\}$, where $\bbeta^{m, j}$is defined in \eqref{eq:GRep}, satisfy 
\begin{align*}
\limsup_{n^m\to\infty}\lambda_j^m q_j^m= l^m<\infty&\hspace{1em}, \nonumber \\
\lim_{n^m\to\infty}\sqrt{\frac{\log(p)}{n^m}}\frac{q_j^m}{\lambda_j^m}=0 \nonumber \\ 
\lim_{n^m\to\infty}\frac{\lambda_j^m \sqrt{q_j^m}}{b^{m,j}_{\min}}=0. 
\end{align*}
\item[({\bf A3})] For $m\in\{\RN{1},\RN{2}\}$ and a given variable $j\in\cV$, define 
the sub-gradient $\tilde\btau^{m,j}$ based on the stationary condition of \eqref{eq:nllasso}
\begin{equation}\label{eq:stationarynl}
\tilde\btau^{m,j}=\frac{1}{n\lambda_j^m}\E\left[\bX_{\backslash j}^{m\top}\left(\bx_j^m - \bX_{\backslash j}^m\tilde\bbeta^{m, j}\right)\right].
\end{equation}
We assume $\tilde\btau^{m,j}$ satisfies $\limsup_{n^m\to\infty}\left\|\tilde\btau^{m,j}_{\backslash\widetilde{ne}_{j}^m}\right\|_\infty \leq 1 - \delta^m$ such as 
\[
\lim_{n^m\to\infty}\sqrt{\frac{\log(p)}{n^m}}\frac{q_j}{\lambda_j^m\delta^m}=0, 
\]
and
\begin{align*}
\lim_{n^m\to\infty}\frac{q_j^m}{\lambda_j^m}\sqrt{\frac{\log(p)}{n^m}}{\left(\min_{k\in \widetilde{ne}_j^m\backslash ne_j^m}\left|\left[\bSigma_{(\widetilde{ne}_j^m,\widetilde{ne}_j^m)}\right]^{-1}\tilde\btau^{m, j}_{\widetilde{ne}_j^m}\right|_k\right)}^{-1} \nonumber \\
=0.
\end{align*}
\end{itemize}

Condition ({\bf A1}) characterizes the data distribution. Combining the first two requirements of ({\bf A2}), for $m\in\{\RN{1},\RN{2}\}$, we get $q_j^m=\smallO((n^m/\log(p))^{1/4})$. 
The third constraint in ({\bf A2}) is the $\bbeta$-min condition, which prevents the signal from being too weak to be detected; this condition may be relaxed to allow the presence of some weak signal variables. 
Note that although our goal is to test the difference in connectivity in two networks, Condition ({\bf A2}) does not require the difference in signal strength to be large under the null hypothesis.
In addition, ({\bf A2}) requires the tuning parameters $\lambda_j^m$ to approach zero at a slower rate than $q_j^m\sqrt{\log(p)/n^m}$, which is the minimum tuning parameter rate for prediction consistency of lasso with Gaussian data \citep[see, e.g.,][]{bickeletal2009}. Since $\lim_{n^m\to\infty}q_j^m\sqrt{\log(p)/n^m}/\lambda_j^m=0$ by ({\bf A2}), condition ({\bf A3}) requires that the tuning parameter $\lambda$ does not converge to any transition points too fast, where some entries of $\tilde\bbeta^{m, j}$ change from zero to nonzero, or vice versa. ({\bf A3}) also requires that $\min_{k\in \widetilde{ne}_j^m\backslash ne_j^m}\big|[\hat\bSigma_{(\widetilde{ne}_j^m,\widetilde{ne}_j^m)}]^{-1}\tilde\btau_{\widetilde{ne}_j^m}^{m,j}\big|_k$ does not converge to zero too fast.

\section{Proof of Proposition~\ref{thm:unbiased}} \label{sec:pfunbiased}
In this section, we prove that conditions ({\bf A1}) and ({\bf A2}) for variable $j\in\cV$ imply 
$$
\lim_{n^\RN{1}\to\infty}\Pr\left[\widehat{ne}_j^{\RN{1}}\supseteq ne_j^{\RN{1}}\right]=1, 
$$
where $\widehat{ne}_j^{\RN{1}}\equiv\supp(\hat\bbeta^{\RN{1},j})$, with $\hat\bbeta^{\RN{1},j}$ defined in \eqref{eq:lasso}. 
The result $\lim_{n^\RN{2}\to\infty}\Pr[\widehat{ne}_j^{\RN{2}}\supseteq ne_j^{\RN{2}}]=1$ can be proved using exactly the same procedure and is thus omitted. Since $\widehat{ne}_j^0\equiv\widehat{ne}_j^{\RN{1}}\cap\widehat{ne}_j^{\RN{2}}$ and $ne_j^0\equiv ne_j^{\RN{1}}\cap ne_j^{\RN{2}}$, the above two results imply
$$
\lim_{n^\RN{1},n^\RN{2}\to\infty}\Pr\left[\widehat{ne}_j^0\supseteq ne_j^0\right]=1, 
$$
Because our proof only concerns population $\RN{1}$, for simplicity, we omit the superscript ``$^{\RN{1}}$'' from the subsequent proofs. 

We first prove Lemma~\ref{lem:resev}, which shows that under ({\bf A1}) and ({\bf A2}), the $(ne_j, q_j, 3)$-restricted eigenvalue condition \citep{bickeletal2009, vandeGeerBuhlmann2009} is satisfied for each $j\in\cV$.

\begin{lemma}\label{lem:resev}
Suppose ({\bf A1}) and ({\bf A2}) hold for variable $j\in\cV$. Suppose $q_j\equiv|ne_j|=|\supp(\bbeta^j)|\geq1$. For all $\bb\in\mathbb{R}^{p-1}$ and any index set $\cI$, such that $|\cI|\leq q_j$, $\|\bb_{\backslash\cI}\|_1\leq 3\|\bb_{\cI}\|_1$ and $\|\bb_{\backslash\mathcal{S}}\|_\infty\leq\min_{j\in(\mathcal{S}\backslash \cI)}|b_j|$, where $\mathcal{S}$ is any index set such that $\mathcal{S}\supseteq \cI$ and $|\mathcal{S}|\leq2q_j$, we have
\begin{align}\label{eq:resevfixed}
\lim_{n\to\infty}\Pr\left[\left\|\bb_{\mathcal{S}}\right\|_2^2\leq\bb^\top\hat\bSigma_{(\backslash j, \backslash j)}\bb\frac{1}{\phi^2}\right]=1,
\end{align}
where $\hat\bSigma_{(\backslash j, \backslash j)}\equiv\bX_{\backslash j}^\top\bX_{\backslash j}/n$ and $\liminf_{n\to\infty}\phi^2=\kappa^2>0$.
\end{lemma}
\begin{proof}
Theorem 1.6 in \citet{Zhou2009resev} shows that with Gaussian data, \eqref{eq:resevfixed} holds if:
\begin{itemize}
	\item[({\bf C1})] $\|\bb_{\mathcal{S}}\|_2^2\leq\bb^\top\bSigma_{(\backslash j, \backslash j)}\bb/\phi^2$ with $\liminf_{n\to\infty}\phi^2>0$;
	\item[({\bf C2})] For any $\bv\in\mathbb{R}^{p-1}$ such that $\|\bv\|_2=1$ and $|\supp(\bv)|\leq q_j$, we have $\bv^\top\bSigma_{(\backslash j, \backslash j)}\bv=\mathcal{O}(1)$;
	\item[({\bf C3})] $\log(p)/n\to0$ and $q_j\log(p/q_j)/n\to0$.
\end{itemize}

We now proceed to show that these three requirements hold. 
\begin{itemize}
\item[({\bf C1})] For any $\bb\in\mathbb{R}^{p-1}$, we have $\bb^\top\bSigma_{(\backslash j, \backslash j)}\bb/\|\bb\|_2^2\geq\phi^2_{\min}[\bSigma_{(\backslash j, \backslash j)}]\geq \phi^2_{\min}[\bSigma]>0$. The second to last inequality is based on the interlacing property of eigenvalues of principal sub-matrices \citep[see, e.g., Theorem 2.1 in][]{Haemers1995}, while the last equality is guaranteed by ({\bf A1}). Thus, for any $\mathcal{S}\subseteq\cV$, we have 
\[
\|\bb_{\mathcal{S}}\|_2^2\leq\|\bb\|_2^2\leq\frac{1}{\phi^2_{\min}\left[\bSigma_{(\backslash j, \backslash j)}\right]}\bb^\top\bSigma_{(\backslash j, \backslash j)}\bb, 
\]
with $\liminf_{n\to\infty}\phi^2_{\min}[\bSigma_{(\backslash j, \backslash j)}]>0$. Thus, ({\bf C1}) is satisfied with $\phi = \phi_{\min}[\bSigma_{(\backslash j, \backslash j)}]$. 
\item[({\bf C2})] For any $\bv\in\mathbb{R}^{p-1}$ such that $\|\bv\|_2=1$ and $|\supp(\bv)|\leq q_j$, we have $\bv^\top\bSigma_{(\backslash j, \backslash j)}\bv\leq\phi^2_{\max}[\bSigma_{(\backslash j, \backslash j)}]\leq\phi^2_{\max}[\bSigma]< \infty$, where the second to last inequality is based on the interlacing property of eigenvalues of principal sub-matrices, and the last equality is guaranteed by ({\bf A1}). 
\item[({\bf C3})] First, combining conditions in ({\bf A2}), we get $q_j\sqrt{q_j\log(p)/n}/b_{\min}^j\to0$, which implies that $q_j\log(p)/n\to0$ and hence $\log(p)/n\to0$. In addition, $q_j\log(p)/n\to0$ implies that $q_j\log(p/q_j)/n\to0$.
\end{itemize}

\end{proof}


We now proceed to prove Proposition~\ref{thm:unbiased} for population $\RN{1}$.

\begin{proof}[Proof of Proposition~\ref{thm:unbiased}] First, if $q_j\equiv|ne_j|=0$, then we trivially have $\widehat{ne}_j\supseteq ne_j$.

If $q_j\equiv|ne_j|\geq1$, we write $\bx_j=\bX_{\backslash j}\bbeta^j+\bepsilon^j$. With Gaussian $\bX$ as required in ({\bf A1}), $\bepsilon^j$ follows a Gaussian distribution.

Theorem 7.2 in \citet{bickeletal2009} shows that with Gaussian design, $(ne_j, q_j, 3)$-restricted eigenvalue condition proved in Lemma~\ref{lem:resev} and $\lambda_j\succsim\sqrt{\log(p)/n}$,
\begin{align}\label{eq:l2bound}
\lim_{n\to\infty}\Pr\left[\left\|\hat\bbeta^j-\bbeta^j\right\|_2\leq\frac{\lambda_j\sqrt{8q_j}}{\phi^2}\right]=1.
\end{align}

In addition, given that $\liminf_{n\to\infty}\phi^2>0$, which is guaranteed by Lemma~\ref{lem:resev}, for $n$ sufficiently large, ({\bf A2}) implies that $b^j_{\min}>3\lambda_j \sqrt{q_j} /\phi^2$. Thus, for any $k$ such that $|\beta^j_k|>0$, 
in the event that 
\begin{align}\label{eq:reqn}
\left\|\hat\bbeta^j-\bbeta^j\right\|_\infty\leq\left\|\hat\bbeta^j-\bbeta^j\right\|_2\leq\frac{\lambda_j\sqrt{8q_j}}{\phi^2},
\end{align}
$|\beta^j_k|>0$ implies $|\hat\beta^j_k|>0$. Therefore, by \eqref{eq:l2bound}, 
\[
\lim_{n\to\infty}\Pr\left[\widehat{ne}_j\supseteq ne_j\right]=1.
\]

\end{proof}


\section{Proof of Proposition~\ref{thm:consistent}} \label{sec:pfconsistent}

Similar to Section~\ref{sec:pfunbiased}, in this section, we prove that ({\bf A1})--({\bf A3}) for some variable $j\in\cV$ imply
\[
\lim_{n^\RN{1}\to\infty}\Pr\left[\widehat{ne}_j^{\RN{1}}= \widetilde{ne}_j^{\RN{1}}\right]= 1.
\]
The counterpart for population $\RN{2}$ can be proved using the same technique. Together these imply 
\[
\lim_{n^\RN{1},n^\RN{2}\to\infty}\Pr\left[\widehat{ne}_j^0= \widetilde{ne}_j^0\right]= 1.
\]
For brevity, we drop the superscript ``$^\RN{1}$'' in the subsequent proofs. We first state and prove lemmas needed for the proof of Proposition~\ref{thm:consistent}. 

\setcounter{theorem}{1}

\begin{lemma}\label{lem:unbiased}
Suppose ({\bf A1}) and ({\bf A2}) for variable $j\in\cV$ hold. Then if $b^j_{\min}> 3\lambda_j \sqrt{ q_j} /\phi^2$, we have $\widetilde{ne}_j\supseteq ne_j$, where $\widetilde{ne}_j\equiv\supp(\tilde\bbeta^j)$ and $ne_j\equiv\supp(\bbeta^j)$. 
\end{lemma}

\begin{proof}
First, if $|ne_j|=0$, we trivially have $\widetilde{ne}_j\supseteq ne_j$.

If $|ne_j|\geq1$, ({\bf A2}) implies that given $\liminf_{n\to\infty}\phi^2>0$, for $n$ sufficiently large, $b^j_{\min}>3\lambda_j\sqrt{q_j} /\phi^2$. 

On the other hand, by Corollary 2.1 in \citet{vandeGeerBuhlmann2009}, ({\bf C1}) in the proof of Lemma~\ref{lem:resev} guarantees that 
\begin{align}\label{comp}
\left\|\tilde\bbeta^j-\bbeta^j\right\|_\infty\leq\left\|\tilde\bbeta^j-\bbeta^j\right\|_2\leq\frac{\lambda_j\sqrt{8q_j}}{\phi^2}. 
\end{align}
Therefore, similar to the proof of Proposition \ref{thm:unbiased}, if $b^j_{\min}> 3\lambda_j \sqrt{q_j} /\phi^2$, for any $k$ such that  $|\beta^j_k|>0$, we have $|\tilde\beta^j_k|>0$ by \eqref{comp}, which implies that $\widetilde{ne}_j\supseteq ne_j$. 
\end{proof}

\begin{lemma}\label{mainlemmaa1}
Suppose ({\bf A1})--({\bf A3}) hold. Then the estimator $\hat\bbeta^j$ defined in \eqref{eq:lasso} satisfies $\|\hat\bbeta^j - \tilde\bbeta^j\|_1 = \mathcal{O}_p\left(q_j\sqrt{\log (p)/n }\right)$.
\end{lemma}

\begin{proof}
Let $Q(\bb) \equiv \|\bx_j-\bX_{\backslash j}\bb\|_2^2/(2n) + \lambda_j\|\bb\|_1$, i.e., $\hat \bbeta^j = \argmin_{\bb \in \mathbb{R}^{p-1}} Q(\bb)$. To prove Lemma~\ref{mainlemmaa1}, we show that for all $\xi > 0$, there exists a constant $m>0$, such that
\begin{equation}\label{1:p0}
\lim_{n\to\infty}\Pr \left[ \inf_{\bb : \left\|\bb\right\|_1 = m}Q\left(\tilde\bbeta^j +\bb q_j\sqrt{\frac{\log(p)}{n}}\,\right) > Q\left(\tilde\bbeta^j\right)\right] =1.
\end{equation}
Because $Q$ is convex, \eqref{1:p0} implies that $\hat \bbeta^j$ lies in the convex region $\{\tilde\bbeta^j + \bb q_j\sqrt{ \log(p)/n} :  \|\bb\|_1 < m\}$ with probability tending to one. Therefore, we have 
\[
\lim_{n\to\infty}\Pr\left[\left\|\hat\bbeta^j - \tilde\bbeta^j\right\|_1 \leq m q_j\sqrt{\frac{\log(p)}{n}}\,\right]=1,
\]
i.e.,  $\|\hat \bbeta^j - \tilde\bbeta^j\|_1 = \mathcal{O}_p\left(q_j\sqrt{\log(p)/n}\right)$.

To prove \eqref{1:p0}, we denote $\bw=\argmin_{\bb : \|\bb\|_1 = m}Q\left(\tilde\bbeta^j +\bb q_j\sqrt{\log(p)/n}\right)$. Expanding terms, we get
\begin{align}
&Q\left(\tilde\bbeta^j +\bw q_j\sqrt{\frac{\log(p)}{n}}\right) -Q \left(\tilde\bbeta^j\right) \nonumber \\
= &\frac{1}{2n}\left\|\left(\bx_j-\bX_{\backslash j}\tilde\bbeta^j\right)-\bX_{\backslash j}\bw q_j\sqrt{\frac{\log(p)}{n}}\right\|_2^2 - \frac{1}{2n}\left\|\bx_j-\bX_{\backslash j}\tilde\bbeta^j\right\|_2^2 \nonumber \\
\quad& + \lambda_j\left\|\tilde\bbeta^j +\bw q_j\sqrt{\frac{\log(p)}{n}}\right\|_1-\lambda_j\left\|\tilde\bbeta^j\right\|_1 \nonumber \\
=&-\frac{q_j\sqrt{\log(p) }}{n^{3/2}}\bw^\top\bX_{\backslash j}^\top\left(\bx_j-\bX_{\backslash j}\tilde\bbeta^j\right) + \frac{q_j^2\log(p)}{2n^2}\bw^\top\bX_{\backslash j}^\top\bX_{\backslash j}\bw \nonumber \\
\quad&+\lambda_j\left\|\tilde\bbeta^j +\bw q_j\sqrt{\frac{\log(p)}{n}}\right\|_1 - \lambda_j\left\|\tilde\bbeta^j\right\|_1. \label{eq:taylor1}
\end{align}

Now, for any $g\neq0,h \in \mathbb{R}$, we have $|g + h|\geq |g| + \mathrm{sign}(g)h$. This is because, 1) if $g$ and $h$ have the same sign, $|g + h| =|g|+|h|= |g| + \mathrm{sign}(h)h= |g| + \mathrm{sign}(g)h$; 2) if they have the opposite signs, $|g+h|=\big||g|-|h|\big|\geq|g|-|h|=|g|-\mathrm{sign}(h)h=|g|+\mathrm{sign}(g)h$; 3) if $h=0$, $|g + h|= |g|=|g| + \mathrm{sign}(g)h$. Thus,
\begin{align}\label{eq:gh}
&\left\|\tilde\bbeta^j + \bw q_j\sqrt{\frac{\log(p)}{n}}\right\|_1 \nonumber \\
=& \left\|\tilde\bbeta_{\backslash\widetilde{ne}_j}^j+\bw_{\backslash\widetilde{ne}_j} q_j\sqrt{\frac{\log(p)}{n}}\right\|_1+\left\|\tilde\bbeta_{\widetilde{ne}_j}^j + \bw_{\widetilde{ne}_j} \sqrt{\frac{\log(p)}{n}} \right\|_1 \nonumber \\
=& q_j\sqrt{\frac{\log(p)}{n}}\left\|\bw_{\backslash\widetilde{ne}_j} \right\|_1+\left\|\tilde\bbeta_{\widetilde{ne}_j}^j + \bw_{\widetilde{ne}_j} q_j\sqrt{\frac{\log(p)}{n}} \right\|_1 \nonumber\\
\geq& q_j\sqrt{\frac{\log(p)}{n}} \left\|\bw_{\backslash\widetilde{ne}_j} \right\|_1 + \left\|\tilde\bbeta_{\widetilde{ne}_j}^j\right\|_1+q_j\sqrt{\frac{\log(p)}{n}} \tilde\btau^{j\top}_{\widetilde{ne}_j}\bw_{\widetilde{ne}_j} \nonumber\\
=&q_j\sqrt{\frac{\log(p)}{n}} \left\|\bw_{\backslash\widetilde{ne}_j} \right\|_1 + \left\|\tilde\bbeta^j\right\|_1+q_j\sqrt{\frac{\log(p)}{n}} \tilde\btau^{j\top}_{\widetilde{ne}_j}\bw_{\widetilde{ne}_j}.
\end{align}
In the second line and the fourth line, we use the fact that $\tilde\bbeta_{\backslash\widetilde{ne}_j}^j={\bf 0}$, and in the third line, we use the fact that $\tilde\btau^j_{\widetilde{ne}_j} =\mathrm{sign}(\tilde\bbeta^j_{\widetilde{ne}_j})$ and the inequality $|g + h|\geq |g| + \mathrm{sign}(g)h$ shown above. Therefore, combining \eqref{eq:taylor1} and \eqref{eq:gh}, we have
\begin{align}
&Q\left(\tilde\bbeta^j +\bw q_j\sqrt{\frac{\log(p)}{n}}\right) -Q \left(\tilde\bbeta^j\right) \nonumber \\ 
\geq&-\frac{q_j\sqrt{\log(p) }}{n^{3/2}}\bw^\top\bX_{\backslash j}^\top\left(\bx_j-\bX_{\backslash j}\tilde\bbeta^j\right) + \frac{q_j^2\log(p)}{2n^2}\bw^\top\bX_{\backslash j}^\top\bX_{\backslash j}\bw \nonumber \\
\quad&+\lambda_j q_j\sqrt{\frac{\log(p)}{ n}}\tilde\btau^{j\top}_{\widetilde{ne}_j}\bw_{\widetilde{ne}_j} +\lambda_j q_j\sqrt{\frac{\log(p)}{ n}} \left\|\bw_{\backslash\widetilde{ne}_j}\right\|_1 \nonumber\\ 
=&-\frac{q_j\sqrt{\log(p) }}{n^{3/2}}\bw^\top\bX_{\backslash j}^\top\left(\bx_j-\bX_{\backslash j}\tilde\bbeta^j\right) + \frac{q_j^2\log(p)}{2n}\bw^\top\hat\bSigma_{(\backslash j, \backslash j)}\bw \nonumber \\
\quad&+\lambda_j q_j\sqrt{\frac{\log(p)}{ n}}\tilde\btau^{j\top}\bw + \lambda_j q_j\sqrt{\frac{\log(p)}{ n}}\left(\left\|\bw_{\backslash\widetilde{ne}_j}\right\|_1 - \tilde\btau^{j\top}_{\backslash\widetilde{ne}_j}\bw_{\backslash\widetilde{ne}_j}\right), \nonumber
\end{align}
where, as before, $\hat\bSigma_{(\backslash j, \backslash j)}=\bX_{\backslash j}^\top\bX_{\backslash j}/n$. Since by ({\bf A3}), $\limsup_{n\to\infty}\|\tilde\btau^j_{\backslash\widetilde{ne}_j}\|_\infty\leq1-\delta$, for $n$ sufficiently large, $\|\tilde\btau^j_{\backslash\widetilde{ne}_j}\|_\infty\leq 1-\delta/2$. Thus, $\tilde\btau^{j\top}_{\backslash\widetilde{ne}_j}\bw_{\backslash\widetilde{ne}_j}\leq|\tilde\btau^{j\top}_{\backslash\widetilde{ne}_j}\bw_{\backslash\widetilde{ne}_j}|\leq\|\tilde\btau^j_{\backslash\widetilde{ne}_j}\|_\infty\|\bw_{\backslash\widetilde{ne}_j}\|_1\leq (1-\delta/2)\|\bw_{\backslash\widetilde{ne}_j}\|_1$, and
\begin{align}
&Q\left(\tilde\bbeta^j +\bw q_j\sqrt{\frac{\log(p)}{n}}\,\right) -Q \left(\tilde\bbeta^j\right) \nonumber \\
\geq&- \frac{q_j\sqrt{ \log(p) }}{n^{3/2}}\bw^\top\left[ \bX_{\backslash j}^\top\left(\bx_j-\bX_{\backslash j}\tilde\bbeta^j\right)  - \lambda_j n\tilde\btau^j\right] \nonumber \\
\quad&+ \frac{ q_j^2\log(p) }{2n}\bw^\top\hat\bSigma_{(\backslash j, \backslash j)}\bw + \lambda_j q_j\frac{\delta}{2}\sqrt{\frac{\log(p)}{ n}} \left\|\bw_{\backslash\widetilde{ne}_j}\right\|_1 \label{eq:taylor3}. 
\end{align}

To bound $\bw^\top[ \bX_{\backslash j}^\top(\bx_j-\bX_{\backslash j}\tilde\bbeta^j)  - \lambda_j n\tilde\btau^j]$ in \eqref{eq:taylor3}, writing $\bx_j=\bX_{\backslash j}\bbeta^j+\bepsilon^j$, we observe 
\begin{align*}
&\frac{1}{n} \bw^\top\left[\bX_{\backslash j}^\top\left(\bx_j-\bX_{\backslash j}\tilde\bbeta^j\right)  - \lambda_j \tilde\btau^j \right] \nonumber \\
=& \bw^\top\left[ \hat\bSigma_{(\backslash j, \backslash j)}\left(\bbeta^j-\tilde\bbeta^j\right)  - \lambda_j\tilde\btau^j +\frac{1}{n}\bX_{\backslash j}^\top\bepsilon^j\right]\nonumber\\
\leq&\left|\bw^\top\left[ \hat\bSigma_{(\backslash j, \backslash j)}\left(\bbeta^j-\tilde\bbeta^j\right)  - \lambda_j \tilde\btau^j +\frac{1}{n}\bX_{\backslash j}^\top\bepsilon^j\right] \right| \nonumber\\
=&\left|\bw^\top\left[ \hat\bSigma_{(\backslash j, \backslash j)}\left(\bbeta^j-\tilde\bbeta^j\right)  - \bSigma_{(\backslash j, \backslash j)}\left(\bbeta^j-\tilde\bbeta^j\right)\right.\right. \\
\quad&+ \left.\left.\bSigma_{(\backslash j, \backslash j)}\left(\bbeta^j-\tilde\bbeta^j\right)- \lambda_j \tilde\btau^j +\frac{1}{n}\bX^\top\bepsilon^j\right]\right|.
\end{align*}
Based on the stationary condition of \eqref{eq:nllasso}, we have $\bSigma_{(\backslash j, \backslash j)}(\bbeta^j-\tilde\bbeta^j)-\lambda_j \tilde\btau^j=\bzero$. Thus,
\begin{align}
&\frac{1}{n} \bw^\top\left[\bX_{\backslash j}^\top\left(\bx_j-\bX_{\backslash j}\tilde\bbeta^j\right)  - \lambda_j \tilde\btau^j \right] \nonumber \\ 
=&\left|\bw^\top\left[ \left(\hat\bSigma_{(\backslash j, \backslash j)}-\bSigma_{(\backslash j, \backslash j)}\right)\left(\bbeta^j-\tilde\bbeta^j\right) +\frac{1}{n}\bX_{\backslash j}^\top\bepsilon^j\right] \right| \nonumber\\
\leq& \left|\bw^\top \left(\hat\bSigma_{(\backslash j, \backslash j)}-\bSigma_{(\backslash j, \backslash j)}\right)\left(\bbeta^j-\tilde\bbeta^j\right)\right| +\left|\frac{1}{n}\bw^\top\bX_{\backslash j}^\top\bepsilon^j \right| \nonumber \\
\leq& \left\|\bw\right\|_1\left\|\left(\hat\bSigma_{(\backslash j, \backslash j)}-\bSigma_{(\backslash j, \backslash j)}\right)\left(\bbeta^j-\tilde\bbeta^j\right)\right\|_\infty + \frac{1}{n}\left\|\bw\big\|_1\big\|\bX_{\backslash j}^\top\bepsilon^j\right\|_\infty \nonumber \\
=& m\left(\left\|\left(\hat\bSigma_{(\backslash j, \backslash j)}-\bSigma_{(\backslash j, \backslash j)}\right)\left(\bbeta^j-\tilde\bbeta^j\right)\right\|_\infty + \frac{1}{n}\left\|\bX_{\backslash j}^\top\bepsilon^j\right\|_\infty\right).\label{ml1}
\end{align}

Based on e.g., Lemma 1 in \citet{Ravikumaretal2011}, assuming ({\bf A1}), we have $\|\hat\bSigma_{(\backslash j, \backslash j)}-\bSigma_{(\backslash j, \backslash j)}\|_{\max}=\mathcal{O}_p(\sqrt{\log(p)/ n})$, where $\|\cdot\|_{\max}$ is the entry-wise infinity norm. In addition, according to Lemma 2.1 in \citet{vandeGeerBuhlmann2009}, with ({\bf C1}) in the proof of Lemma~\ref{lem:resev}, we have $\|\bbeta^j-\tilde\bbeta^j\|_1=\mathcal{O}(\lambda_j q_j)$. Thus,  
\begin{align}\label{eq:sigbbound}
&\left\|\left(\hat\bSigma_{(\backslash j, \backslash j)}-\bSigma_{(\backslash j, \backslash j)}\right)\left(\bbeta^j-\tilde\bbeta^j\right)\right\|_\infty \nonumber \\
\leq&\left\|\hat\bSigma_{(\backslash j, \backslash j)}-\bSigma_{(\backslash j, \backslash j)}\right\|_{\max}\left\|\bbeta^j-\tilde\bbeta^j\right\|_1\nonumber \\
=&\mathcal{O}_p\left(\lambda_j q_j\sqrt{\frac{\log(p)}{n}}\right).
\end{align}
Since $\lambda_j q_j\to l<\infty$ in ({\bf A3}), we have $\|(\hat\bSigma_{(\backslash j, \backslash j)}-\bSigma_{(\backslash j, \backslash j)})(\bbeta^j-\tilde\bbeta^j)\|_\infty=\mathcal{O}_p(\sqrt{\log(p)/n})$. Based on a well-known result on Gaussian random variables, we also have $\|\bX_{\backslash j}^\top\bepsilon^j\|_\infty/n=\mathcal{O}_p(\sqrt{\log(p)/ n})$. Thus, we obtain 
\begin{align}\label{eq:bound1}
\frac{1}{mn}\bw^\top\left[ \bX_{\backslash j}^\top\left(\bx_j-\bX_{\backslash j}\tilde\bbeta^j\,\right)  - \lambda_j n\tilde\btau^j\right]=\mathcal{O}_p\left(\sqrt{\frac{\log(p)}{n}}\right).
\end{align}

To bound the other term 
\[
q_j^2\log(p) \bw^\top\hat\bSigma_{(\backslash j, \backslash j)}\bw/2n + \lambda_j\delta q_j\sqrt{\log(p)}\|\bw_{\backslash\widetilde{ne}_j} \|_1/(2\sqrt{n})
\]
in \eqref{eq:taylor3}, consider two cases: 1) $\widetilde{ne}_j=\emptyset$ and 2) $\widetilde{ne}_j\neq\emptyset$.

If 1) $\widetilde{ne}_j=\emptyset$,
\begin{align*}
&\frac{q_j^2\log(p)}{2n} \bw^\top\hat\bSigma_{(\backslash j, \backslash j)}\bw + \frac{\lambda_j\delta q_j}{2} \sqrt{\frac{\log(p)}{n}}\left\|\bw_{\backslash\widetilde{ne}_j} \right\|_1 \\
\geq& \frac{\lambda_j\delta q_j}{2} \sqrt{\frac{\log(p)}{n}}\left\|\bw\right\|_1\\
=&\frac{\lambda_j\delta q_j}{2}\sqrt{\frac{\log(p)}{n}} m.
\end{align*}
The first inequality is due to the positive semi-definitiveness of $\hat\bSigma_{(\backslash j, \backslash j)}$. Hence, 
\begin{align}
&Q\left(\tilde\bbeta^j +\bw q_j\sqrt{\frac{\log(p)}{n}}\right) -Q \left(\tilde\bbeta^j\right) \nonumber \\
\geq& mq_j\sqrt{\frac{\log(p)}{n}} \left(\lambda_j\frac{\delta}{2} - \frac{1}{mn}\bw^\top\left[ \bX_{\backslash j}^\top\left(\bx_j-\bX_{\backslash j}\tilde\bbeta^j\right)-\lambda_jn\tilde\btau^j\right]\right).
\end{align}
By \eqref{eq:bound1}, 
and given that $\sqrt{\log(p)/n}/(\lambda_j\delta)\to0$ by ({\bf A3}), for any $m>0$
\[
Q\left(\tilde\bbeta^j +\bw q_j\sqrt{\frac{\log(p)}{n}}\,\right) >Q \left(\tilde\bbeta^j\right)
\]
with high probability, which implies that \eqref{1:p0} holds.

If 2) $\widetilde{ne}_j\neq\emptyset$, i.e., $q_j\geq1$, we further consider two cases: i) $\|\bw_{\backslash\widetilde{ne}_j}\|_1> 3\|\bw_{\widetilde{ne}_j}\|_1$, and ii) $\|\bw_{\backslash ne_j}\|_1\leq 3\|\bw_{ne_j}\|_1$. Note that these two cases are not mutually exclusive. However, we proved in Lemma~\ref{lem:unbiased} that if $b^j_{\min}> 3\lambda_j\sqrt{q_j}/ \phi^2$, which happens when $n$ is sufficiently large, then $\widetilde{ne}_j\supseteq ne_j$. Thus, for any $\bw$, we have $\|\bw_{\backslash\widetilde{ne}_j}\|_1\leq \|\bw_{\backslash{ne}_j}\|_1$ and $\|\bw_{ne_j}\|_1\leq\|\bw_{\widetilde{ne}_j}\|_1$. Therefore, although the two cases i) $\|\bw_{\backslash\widetilde{ne}_j}\|_1> 3\|\bw_{\widetilde{ne}_j}\|_1$, and ii) $\|\bw_{\backslash ne_j}\|_1\leq 3\|\bw_{ne_j}\|_1$ are not mutually exclusive, they cover all possibilities.

If i) $\|\bw_{\backslash\widetilde{ne}_j}\|_1> 3\|\bw_{\widetilde{ne}_j}\|_1$, because $\|\bw\|_1=\|\bw_{\backslash\widetilde{ne}_j}\|_1+\|\bw_{\widetilde{ne}_j}\|_1=m$, $\|\bw_{\backslash\widetilde{ne}_j}\|_1>3m/4$, and
\begin{align}
&\frac{ \log(p) }{2n}\bw^\top\hat\bSigma_{(\backslash j, \backslash j)}\bw + \lambda_j \frac{\delta}{2}\sqrt{\frac{\log(p)}{ n}} \left\|\bw_{\backslash\widetilde{ne}_j}\right\|_1 \nonumber \\
\geq& \lambda_j \frac{\delta}{2}\sqrt{\frac{\log(p)}{ n}} \left\|\bw_{\backslash\widetilde{ne}_j}\right\|_1\nonumber \\
>& \lambda_j \frac{\delta}{2}\sqrt{\frac{\log(p)}{ n}}\frac{3m}{4}.\label{case1}
\end{align}
Combining \eqref{eq:taylor3}, \eqref{ml1} and \eqref{case1}, we get
\begin{align}
&Q\left(\tilde\bbeta^j +\bw q_j\sqrt{\frac{\log(p)}{n}}\right) -Q \left(\tilde\bbeta^j\right) \nonumber \\
\geq&  m q_j\sqrt{\frac{\log(p)}{n}} \left(\lambda_j\frac{3\delta}{8} - \frac{1}{mn}\bw^\top\left[ \bX_{\backslash j}^\top\left(\bx_j-\bX_{\backslash j}\tilde\bbeta^j\right)-\lambda_jn\tilde\btau^j\right]\right).
\end{align}
Because $\sqrt{\log(p)/n}/(\lambda_j\delta)\to0$ by ({\bf A3}) and $\bw^\top[ \bX_{\backslash j}^\top(\bx_j-\bX_{\backslash j}\tilde\bbeta^j)-\lambda_jn\tilde\btau^j]/(mn) = \mathcal{O}_p\left(\sqrt{\log(p)/n}\right)$ by \eqref{eq:bound1}, with any $m>0$ and $n$ sufficiently large, we have
\[
Q\left(\tilde\bbeta^j +\bw q_j\sqrt{\frac{\log(p)}{n}}\,\right) >Q \left(\tilde\bbeta^j\right),
\]
and hence \eqref{1:p0} holds.

On the other hand, if ii) $\|\bw_{\backslash ne_j}\|_1\leq 3\|\bw_{ne_j}\|_1$, because $\|\bw\|_1=m$, we have  $|\bw_{ne_j}|\geq m/4$. 
Let $\mathcal{S}= ne_j\cup\{l\}$ where $l=\argmax_{j:j\notin  ne_j} |w_j|$. Then $\mathcal{S}\supseteq ne_j$, $|\mathcal{S}|=q_j+1\leq2q_j$ and $\|\bw_{\backslash\mathcal{S}}\|_\infty\leq\min_{j\in\mathcal{S}\backslash ne_j}|w_j|$. Hence, the $(ne_j, q_j, 3)$-restricted eigenvalue condition in Lemma~\ref{lem:resev} implies that, with probability tending to one, as $n$ approaches infinity, 
\begin{align}
&q_j^2\frac{\log(p) }{2n}\bw^\top\hat\bSigma_{(\backslash j, \backslash j)}\bw + \lambda_j q_j\frac{\delta}{2}\sqrt{\frac{\log(p)}{ n}} \left\|\bw_{\backslash\widetilde{ne}_j}\right\|_1 \nonumber \\
\geq& q_j^2\frac{\log(p) }{2n}\bw^\top\hat\bSigma_{(\backslash j, \backslash j)}\bw \nonumber\\
\geq& q_j^2\frac{\log(p) \phi^2}{2n}\left\|\bw_{\mathcal{S}}\right\|_2^2 \nonumber \\
\geq& q_j^2\frac{\log(p) \phi^2}{2n}\left\|\bw_{ne_j}\right\|_2^2 \nonumber \\
\geq& q_j\frac{\log(p) \phi^2}{2n}\left\|\bw_{ne_j}\right\|_1^2 \nonumber \\
\geq&\frac{q_j\log(p) \phi^2}{2n}\frac{m^2}{16}. \label{case2}
\end{align}
Thus, combining \eqref{eq:taylor3}, \eqref{ml1} and \eqref{case2}, we find that for $n$ sufficiently large, $\phi^2\geq\kappa^2/2$, and
\begin{align*}
&Q\left(\tilde\bbeta^j +\bw q_j\sqrt{\frac{\log(p)}{n}}\right) -Q \left(\tilde\bbeta^j\right)\\
\geq& q_jm\sqrt{\frac{\log(p)}{n}} \left(\frac{\kappa^2}{128}\sqrt{\frac{\log(p)}{n}}m \right. \\
\quad&\left.- \frac{1}{mn}\bw^\top\left[ \bX_{\backslash j}^\top\left(\bx_j-\bX_{\backslash j}\tilde\bbeta^j\right)-\lambda_jn\tilde\btau^j\right]\right).
\end{align*}
Since $\bw^\top[ \bX_{\backslash j}^\top(\bx_j-\bX_{\backslash j}\tilde\bbeta^j)-\lambda_jn\tilde\btau^j]/(mn) =\mathcal{O}_p(\sqrt{\log(p)/n})$, we can choose $m$ to be sufficiently large, not depending on $n$, such that \eqref{1:p0} holds. 
\end{proof}


\begin{lemma}\label{mainlemmaa2}
Suppose ({\bf A2}) and ({\bf A3}) hold. For $\widetilde{ne}_j \neq\emptyset$, we have 
\begin{align}
\sqrt{\frac{\log(p)}{n}}\frac{q_j}{\tilde{b}^j_{\min}}\to0,
\end{align}
where $\tilde{b}^j_{\min}\equiv\min\left\{|\tilde\beta^j_k|:\tilde\beta^j_k\neq0\right\}$.
\end{lemma}
\begin{proof}
We show that $q_j\sqrt{\log(p)/n}/|\tilde\beta^j_k|\to0$ for any $k\in \widetilde{ne}_j$ by considering two cases: 1) for $k\in ne_j$, and 2) for $k\in \widetilde{ne}_j\backslash ne_j$.

1) For any $k\in ne_j$, and for $n$ sufficiently large, Lemma~\ref{lem:unbiased} indicates that
\[
\left\|\tilde\bbeta^j-\bbeta^j\right\|_\infty\leq\left\|\tilde\bbeta^j-\bbeta^j\right\|_2\leq\frac{\lambda_j\sqrt{8q_j}}{\phi^2}.
\]
Now, by ({\bf A2}), for any $k\in ne_j$, $|\beta_k^j|>3\lambda_j\sqrt{q_j}/\phi^2$,  with a sufficiently large $n$. Hence, for any $k\in ne_j$, $|\tilde\beta_k^j|>(3-2\sqrt{2})\lambda_j\sqrt{q_j}/\phi^2$. Therefore, given the rates in ({\bf A2}),
\begin{align*}
0<\sqrt{\frac{\log(p)}{n}}\frac{q_j}{\left|\tilde\beta_k^j\right|}<\sqrt{\frac{\log(p)}{n}}\frac{q_j}\lambda_j\cdot\frac{\phi^2}{\left(3-2\sqrt{2}\right)\sqrt{q_j}}\to0.
\end{align*}

If $\widetilde{ne}_j= ne_j$, then our proof is complete. Otherwise, 2) for $k\in \widetilde{ne}_j\backslash ne_j$, consider the stationary condition of \eqref{eq:nllasso},
\begin{align}
n\lambda_j\tilde\btau^j_{ne_j}&=\E\left[\bX_{\widetilde{ne}_j}^\top\bX\right]\left(\bbeta^j-\tilde\bbeta^j\right) \nonumber \\
&= \E\left[\bX_{\widetilde{ne}_j}^\top\bX_{\widetilde{ne}_j}\right]\left(\bbeta^j_{\widetilde{ne}_j}-\tilde\bbeta^j_{\widetilde{ne}_j}\right).
\end{align}
The second equality holds because based on Lemma~\ref{lem:unbiased}, $\widetilde{ne}_j\supseteq ne_j$, i.e., $\bbeta^j_{\backslash\widetilde{ne}_j}=\tilde\bbeta^j_{\backslash\widetilde{ne}_j}=\bzero$. Rearranging terms,
\begin{align}
\tilde\bbeta^j_{\widetilde{ne}_j}=\bbeta^j_{\widetilde{ne}_j}-n\lambda_j\E\left[\bX_{\widetilde{ne}_j}^\top\bX_{\widetilde{ne}_j}\right]^{-1}\tilde\btau^j_{\widetilde{ne}_j}.
\end{align}
Recall that for any $k\in \widetilde{ne}_j\backslash ne_j$, $\beta^j_k=0$. Thus, for any $k\in \widetilde{ne}_j\backslash ne_j$,
\begin{align}\label{eq:bminkeyinequality}
\left|\tilde\beta^j_k\right|=\left|n\lambda_j\E\left[\bX_{\widetilde{ne}_j}^\top\bX_{\widetilde{ne}_j}\right]^{-1}\tilde\btau^j_{\widetilde{ne}_j}\right|_k =\lambda_j\left|\left[\bSigma_{(\widetilde{ne}_j,\widetilde{ne}_j)}\right]^{-1}\tilde\btau^j_{\widetilde{ne}_j}\right|_k.
\end{align}

By ({\bf A3}), we have
\[
\sqrt{\frac{\log(p)}{n}}\frac{q_j}\lambda_j\left(\min_{k\in \widetilde{ne}_j\backslash ne_j}\left|\left[\bSigma_{(\widetilde{ne}_j,\widetilde{ne}_j)}\right]^{-1}\tilde\btau^j_{\widetilde{ne}_j}\right|_k\right)^{-1}\to0,
\]
which means for any $k\in \widetilde{ne}_j\backslash ne_j$, we have 
\[
\sqrt{\frac{\log(p)}{n}}\frac{q_j}{\left|\tilde\beta_k^j\right|}\to0.
\]

\end{proof}

Now we proceed to prove Proposition~\ref{thm:consistent}.

\begin{proof}[Proof of Proposition~\ref{thm:consistent}]
We first note that according to the stationary conditions of \eqref{eq:nllasso} and \eqref{eq:lasso}, respectively, we have
\begin{align}
\label{eq:stationarya}
\tilde\btau^j&=\frac{1}{\lambda_j}\bSigma_{(\backslash j, \backslash j)}\left(\bbeta^j-\tilde\bbeta^j\right), \\
\label{eq:stationaryb}
\hat\btau^j&=\frac{1}{n\lambda_j}\bX_{\backslash j}^\top\left(\bx_j-\bX_{\backslash j}\hat\bbeta^j\right).
\end{align}
Writing $\bx_j=\bX_{\backslash j}\bbeta^j+\bepsilon^j$, \eqref{eq:stationaryb} gives
\begin{align}\label{eq:stationaryb2}
\hat\btau^j=\frac{1}{\lambda_j}\hat\bSigma_{(\backslash j, \backslash j)}\left(\bbeta^j-\hat\bbeta^j\right)+\frac{1}{n\lambda_j}\bX_{\backslash j}^\top\bepsilon^j
\end{align}
where $\hat\bSigma_{(\backslash j, \backslash j)}=\bX_{\backslash j}^\top\bX_{\backslash j} / n$. Combining \eqref{eq:stationarya} and \eqref{eq:stationaryb2}, 
\begin{align}\label{eq:tau}
\hat\btau^j-\tilde\btau^j &= \frac{1}{\lambda_j}\left(\hat\bSigma_{(\backslash j, \backslash j)}\left[\bbeta^j-\hat\bbeta^j\right]-\bSigma_{(\backslash j, \backslash j)}\left[\bbeta^j-\tilde\bbeta^j \right]\right) + \frac{1}{n\lambda_j}\bX_{\backslash j}^\top\bepsilon^j \nonumber \\
&=\frac{1}{\lambda_j}\left(\hat\bSigma_{(\backslash j, \backslash j)}\left[\bbeta^j-\hat\bbeta^j\right]-\bSigma_{(\backslash j, \backslash j)}\left[\bbeta^j-\hat\bbeta^j\right]\right. \nonumber\\
&\quad+\left.\bSigma_{(\backslash j, \backslash j)}\left[\bbeta^j-\hat\bbeta^j\right] -\bSigma_{(\backslash j, \backslash j)}\left[\bbeta^j-\tilde\bbeta^j\right] \right) + \frac{1}{n\lambda_j}\bX_{\backslash j}^\top\bepsilon^j \nonumber \\
&=\frac{1}{\lambda_j}\left(\hat\bSigma_{(\backslash j, \backslash j)}-\bSigma_{(\backslash j, \backslash j)}\right)\left(\bbeta^j-\hat\bbeta^j\right) \nonumber \\
&\quad+\frac{1}{\lambda_j}\bSigma_{(\backslash j, \backslash j)}\left(\tilde\bbeta^j-\hat\bbeta^j\right) + \frac{1}{n\lambda_j}\bX_{\backslash j}^\top\bepsilon^j.
\end{align}

We now bound all three terms on the right hand side of \eqref{eq:tau}. First, by the Gaussianity of the data, 
$$
\frac{1}{n\lambda_j}\left\|\bX_{\backslash j}^\top\bepsilon^j\right\|_\infty=\mathcal{O}_p\left(\frac{1}{\lambda_j}\sqrt{\frac{\log(p)}{n}}\right).
$$ 
In addition, we proved in Lemma~\ref{mainlemmaa1} that $\|\hat \bbeta^j - \tilde\bbeta^j\|_1 = \mathcal{O}_p\left(q_j\sqrt{\log(p)/n }\right)$. Thus, because $\left\|\bSigma_{(\backslash j, \backslash j)}\right\|_{\max} = 1$,
\begin{align*}
\frac{1}{\lambda_j}\left\|\bSigma_{(\backslash j, \backslash j)}\left(\hat \bbeta^j - \tilde\bbeta^j\right)\right\|_\infty&\leq\frac{1}{\lambda_j}\left\|\bSigma_{(\backslash j, \backslash j)}\right\|_{\max}\left\|\hat \bbeta^j - \tilde\bbeta^j\right\|_1\\
&=\mathcal{O}_p\left(\frac{q_j}{\lambda_j}\sqrt{\frac{\log(p)}{n}}\right).
\end{align*}
Finally, based on Theorem 7.2 in \citet{bickeletal2009}, Lemma~\ref{lem:resev} imply that $\|\bbeta^j-\hat\bbeta^j\|_1=\mathcal{O}_p(\lambda_j q_j)$. Thus,
\begin{align*}
&\quad\left\|\left(\hat\bSigma_{(\backslash j, \backslash j)}-\bSigma_{(\backslash j, \backslash j)}\right)\left(\bbeta^j-\hat\bbeta^j\right)\right\|_\infty\\
&\leq\left\|\hat\bSigma_{(\backslash j, \backslash j)}-\bSigma_{(\backslash j, \backslash j)}\right\|_{\max}\left\|\bbeta^j-\hat\bbeta^j\right\|_1 \\
&=\mathcal{O}_p\left(\sqrt{\frac{\log(p)}{n}}\lambda_j q_j\right)=\mathcal{O}_p\left(\sqrt{\frac{\log(p)}{n}}\right).
\end{align*}
Thus, \eqref{eq:tau} shows that $\|\tilde\btau^j-\hat\btau^j\|_\infty=\mathcal{O}_p\left(q_j\sqrt{\log(p)/n}/\lambda_j\right)$. By ({\bf A3}), we have that $\limsup_{n\to\infty}\|\tilde\btau^j_{\backslash\widetilde{ne}_j}\|_\infty \leq 1 - \delta$ for $q_j\sqrt{\log(p)/n}/(\lambda_j\delta)\to0$, and hence $\lim_{n\to\infty}\Pr[\|\hat\btau^j_{\backslash\widetilde{ne}_j}\|_\infty<1]=1$. Thus, 
\begin{align}\label{side2}
\lim_{n\to\infty}\Pr\left[ \widetilde{ne}_j\supseteq\hat ne_j\right]=1. 
\end{align}

On the other hand, if $\widetilde{ne}_j= \emptyset$, $\widetilde{ne}_j\subseteq \hat ne_j$. If $\widetilde{ne}_j\neq \emptyset$, by Lemma~\ref{mainlemmaa2},
\begin{align}\label{eq:cont1}
\sqrt{\frac{\log(p)}{n}}\frac{q_j}{\tilde{b}_{\min}^j}\to0.
\end{align}
Lemma~\ref{mainlemmaa1} shows that $\|\hat \bbeta^j - \tilde\bbeta^j\|_1 = \mathcal{O}_p\left(q_j\sqrt{\log(p)/n }\right)$, i.e., there exists a constant $C>0$ such that  
\begin{align}\label{eq:cont2}
\lim_{n\to\infty}\Pr\left[\left\|\hat \bbeta^j - \tilde\bbeta^j\right\|_\infty > Cq_j\sqrt{\frac{\log(p)}{n} }\,\right] =0.
\end{align}
Based on \eqref{eq:cont1}, for $n$ sufficiently large, $\tilde{b}_{\min}^j > 2Cq_j\sqrt{\log(p)/n}$. Thus, combining \eqref{eq:cont1} and \eqref{eq:cont2}, for $n$ sufficiently large, whenever $\tilde\beta_k^j \neq 0$, we have $|\tilde\beta_k^j|  > 2Cq_j\sqrt{\log(p)/n}$ and hence $\lim_{n\to\infty}\Pr\left[|\hat\beta_k^j|  > 0\right]=1$. Therefore
\begin{align}\label{side1}
\lim_{n\to\infty}\Pr\left[\widetilde{ne}_j\subseteq\hat ne_j\right]=1,
\end{align}
which completes the proof. 
\end{proof}


\section{Details for the Example in Figure~\ref{fig:quanvsqual}} \label{sec:quanvsqual}

Since $\bX^{\RN{1}}$ is normally distributed and 
\[
\bOmega^{\RN{1}}=\begin{bmatrix}
    1       & 0.5 & 0.5 \\
    0.5    & 1    & 0.5 \\
    0.5    & 0.5    & 1
\end{bmatrix},
\]
we have $\bX^{\RN{1}}\sim_{i.i.d.}\mathcal{N}_3(\bzero, \bSigma^{\RN{1}})$, where
\[
\bSigma^{\RN{1}}=\left(\bOmega^{\RN{1}}\right)^{-1}=\begin{bmatrix}
    1.5       & -0.5 & -0.5 \\
    -0.5    & 1.5    & -0.5 \\
    -0.5    & -0.5    & 1.5
\end{bmatrix}.
\]
In population $\RN{2}$, we have $\bx_1^{\RN{2}}=_d\bx_1^{\RN{1}}$, $\bx_2^{\RN{2}}=_d\bx_2^{\RN{1}}$ and $\bx_3^{\RN{2}}\independent\bx_{\{1,2\}}^{\RN{2}}$. Thus, $\bX^{\RN{2}}\sim_{i.i.d.}\mathcal{N}_3(\bzero, \bSigma^{\RN{2}})$, where
\[
\bSigma^{\RN{2}}=\begin{bmatrix}
    1.5       & -0.5 & 0 \\
    -0.5    & 1.5    & 0 \\
    0    & 0    & \Var\left(\bx_3^{\RN{2}}\right)
\end{bmatrix},
\]
which implies that 
\[
\bOmega^{\RN{2}}=\left(\bSigma^{\RN{2}}\right)^{-1}=\begin{bmatrix}
    0.75       & 0.25 & 0 \\
    0.25    & 0.75    & 0 \\ 
    0    & 0    & 1/\Var\left(\bx_3^{\RN{2}}\right).
\end{bmatrix}.
\]

\section{Detail of the Brain Imaging Data Collection and Processing}\label{sec:data}

MRI scans were completed on a 3T Phillips Achieva with a 32-channel head coil.  Each imaging session lasted ~40 minutes and included a 1mm isotropic MPRAGE (5:13), 1mm isotropic 3D T2-weighted image (5:22), 1mm isotropic 3D T2-Star (3:41), 2D FLAIR collected at an in-plane resolution of 1 x 1mm with a slice thickness of 4mm, no gaps (2:56), 3mm isotropic BOLD image for resting-state fMRI (6:59), and a 32 direction 2mm isotropic diffusion sequence acquired with reverse polarity (A-P, P-A), b=1000 sec/mm2, and 6 non-diffusion weighted images for diffusion tensor imaging (DTI) analysis (each 6:39).  The DTI data was then post-processed by the collaborative group by the following methods before the data was then utilized for the current analysis.  Briefly, the first portion of the pipeline uses TORTOISE - Tolerably Obsessive Registration and Tensor Optimization Indolent Software Ensemble \citep{pierpaoli2010tortoise}.  For reverse polarity data, each DWI acquisition both A-P and P-A is initially run through DiffPrep \citep{oishi2009atlas, zhang2010atlas} in TORTOISE for susceptibility distortion correction, motion correction, eddy current correction, and registration to 3D high resolution structural image. For EPI distortion correction, the diffusion images were registered to the 1mm isotropic T2 image using non-linear b-splines.  Eddy current and motion distortion were corrected using standard affine transformations followed by re-orientation of the b-matrix for the rotational aspect of the rigid body motion.  Following DiffPrep, the output images from both the A-P and P-A DWI acquisitions were then sent through Diffeomorphic Registration for Blip-Up Blip-Down Diffusion Imaging \citep[DR-BUDDI,][]{irfanoglu2015} in TORTOISE for further EPI distortion and eddy current correction that can be completed with diffusion data that has been collected with reverse polarity. This step combines the reverse polarity imaging data creating a single, cleaned, DWI data set that is then sent through DiffCalc \citep{pierpaoli2010tortoise,koay2006,koay2009, basser1994,mangin2002,chang2005,chang2012,rohde2005} in TORTOISE. This step completes the tensor estimation using the robust estimation of tensors by outlier rejection (RESTORE)10 approach.  Following tensor estimation, a variety of DTI metrics can be derived. For this study, we specifically focused on fractional anisotropy (FA) as our main metric.   

Following this post-processing in TORTOISE, 3D image stacks for MD and FA were introduced into DTIstudio \citep{zhang2010atlas,oishi2009atlas} for segmentation of the DTI atlas \citep{mori2005mri} on to each participants DTI data set in `patient space' through the Diffeomap program in DTIstudio using both linear and non-linear transformations.  This is a semi-automated process that allows for the extraction of DTI metrics within each 3D-atlas-based region of interest providing a comprehensive sampling throughout the entire brain into 189 regions including ventricular space.  For this study, selection of regions was limited to regions of white matter as the main hypothesis regarding DTI was that there would be reductions in white matter integrity observed with FA related to brain injury.  This reduced the number of regions used for further analysis down to 78.  To select only white matter, FA images were then threshold at 0.2 or greater, and this final 3D segmentation was then applied to all other co-registered DTI metrics and the data within each DTI metric for the 78 regions of interest was extracted in ROIeditor for further analysis.

\newpage
\bibliographystyle{apalike}
\bibliography{DiffNet}
\end{document}